\documentclass{vldb}

\usepackage{graphicx}
\usepackage{balance}  
\usepackage{subfig}

\usepackage{here}
\usepackage{algorithm}
\usepackage[noend]{algpseudocode}

\def\WR{\stackrel{wr}{\to}}
\def\VORW{\stackrel{\ll(rw)}{\to}}
\def\VOWW{\stackrel{\ll(ww)}{\to}}

\newtheorem{counter}{counter}
\newtheorem{definition}[counter]{Definition}
\newtheorem{observation}[counter]{Observation}
\newtheorem{proposition}[counter]{Proposition}
\newtheorem{theorem}[counter]{Theorem}

\newtheorem{lemma}[counter]{Lemma}

\begin{document}


\title{NWR: Rethinking Thomas Write Rule for Omittable Write Operations}



%
%
%
%

\numberofauthors{1} 

\author{
%
%
       \alignauthor
       Sho Nakazono\footnotemark[1],
       Hiroyuki Uchiyama\footnotemark[1],
       Yasuhiro Fujiwara\footnotemark[2],
       Yasuhiro Nakamura\footnotemark[4],
       Hideyuki Kawashima\footnotemark[5]
\\
       \affaddr{
              \footnotemark[1]\ \,NTT Software Innovation Center\\
              \footnotemark[2]\ \,NTT Communication Science Laboratories\\
              \footnotemark[4]\ \,Graduate School of Systems and Information Engineering, University of Tsukuba\\
              \footnotemark[5]\ \,Faculty of Environment and Information Studies, Keio University\\
       }
}



\maketitle
\renewcommand\thefootnote{\fnsymbol{footnote}}	
\footnotetext[4]{Now works at Yahoo Japan Corporation}
\renewcommand\thefootnote{\arabic{footnote}}

\begin{abstract}
    Concurrency control protocols are the key to scaling current DBMS performances.
They efficiently interleave read and write operations in transactions, but occasionally they restrict concurrency by using coordination such as exclusive lockings.
Although exclusive lockings ensure the correctness of DBMS, it incurs serious performance penalties on multi-core environments.
In particular, existing protocols generally suffer from emerging highly write contended workloads, since they use innumerable lockings for write operations.
In this paper, we rethink the Thomas write rule (TWR), which allows the timestamp ordering (T/O) protocol to omit write operations without any lockings.
We formalize the notion of omitting and decouple it from the T/O protocol implementation, in order to define a new rule named non-visible write rule (NWR).
When the rules of NWR are satisfied, any protocol can in theory generate omittable write operations with preserving the correctness without any lockings.
In the experiments, we implement three NWR-extended protocols: Silo+NWR, TicToc+NWR, and MVTO+NWR.
Experimental results demonstrate the efficiency and the low-overhead property of the extended protocols.
We confirm that NWR-extended protocols achieve more than 11x faster than the originals in the best case of highly write contended YCSB-A and comparable performance with the originals in the other workloads.

\end{abstract}

\section{Introduction}
Thanks to hundreds of CPU cores and terabytes of RAM on the market, DBMS’s performance bottleneck has shifted from storage I/O to the concurrent processing \cite{Harizopoulos2008OLTPThere}, and thus concurrency control protocols are the key to scaling DBMS performance.
They efficiently interleave read and write operations in transactions; however, they occasionally restrict concurrency by using coordination such as exclusive locking.
Minimizing the number of lockings is a notoriously difficult problem.
Although lockings incurs serious performance penalties, concurrency control protocols cannot avoid to use lockings for preserving \textit{strict serializability} \cite{Herlihy1990Linearizability:Objects} and \textit{recoverability} \cite{Hadzilacos1988ASystems}, the important correctness properties of DBMS.
A number of studies \cite{Lim2017Cicada:Transactions,Kim2016Ermia:Workloads, Yu2016Tictoc:Control,Wang2017EfficientlySerializable,Tu2013SpeedyDatabases,Durner2019NoFalseNegatives,Ding2018ImprovingOptimistic} have addressed this trade-off and succeeded in reducing the number of lockings for read contentions, especially on current industrial standard benchmarks (e.g., TPC-C).

However, recent protocols suffer from emerging workloads on IoT applications and cloud services since their workloads exhibit tremendous write contentions; they frequently collect data from various sources outside DBMS such as sensors, servers, browsers, and mobile devices, and they blindly update shared states simultaneously.
Write contentions in such workloads severely limit the scalability on multi-core processors.
This is because common concurrency control protocols use one exclusive locking per one write operation, and process contended transactions as serial execution.
In fact, it has been known that the performance of recent protocols degrades as the level of write contention increases \cite{Wu2017AnControl, Yu2014StaringCores, Kim2016Ermia:Workloads, Fan:2019:OVG:3342263.3360357}.

Is it possible to reduce the number of lockings for write contentions?
The answer is \textbf{yes}; we show the answer to the question in this paper.
For an intuitive understanding of our approach, we revisit the Thomas write rule (TWR) \cite{Thomas1977ABases} proposed in 1977.
With TWR, we can execute multiple write operations with a single locking.
TWR is a well-known but antiquated optimization rule for the single-version timestamp ordering (T/O) protocol.
If a write operation encounters another one that has a larger timestamp, TWR regards the former is obsolete and \textit{``omittable''}; the execution and locking of the former are unnecessary.
Even when tremendous write operations for the same data item are running concurrently, we need only a single locking for the single operation with the largest timestamp.

The limitations of TWR are, however, widely known.
First, T/O with TWR incurs unacceptable results for users.
This is because T/O with TWR does not ensure both strict serializability and recoverability as the correctness properties.
Second, TWR is not applicable to the other protocols.
The rule of TWR assumes to use per-transaction, totally-ordered, and monotonically-increasing timestamps; however, recent protocols have not always such timestamps.
T/O with TWR is not in the choices of current DBMSs nowadays, and TWR has not been well studied to apply the concept of omittable for recent protocols.

\textbf{Our goal is to construct an easily-applicable rule to omit write operations with preserving the correctness.}
In this paper, we rethink the formal aspects of omittability.
Although TWR is recognized as the rule for the single-version T/O protocol, we analyze it in terms of the theory of multi-version protocols.
The concept of \textit{version orders} \cite{Bernstein1983MultiversionAlgorithms} in multiversion serializability theory is essential to provide both the theoretical analysis and the applicable approach.
For theoretical analysis, version orders enable us to formalize the definition of omittability without the T/O's timestamps.
As a result, we found that an omittable write operation preserves the correctness only if there exists a version order reverse to the operation order.
Version orders also play a crucial role in the applicability.
It derives the idea that any protocol can be extended to omit write operations by \textbf{validating an additional (reverse) version order}.
Generally, concurrency control protocols generate and validate only a single version order, which is tightly coupled to the implementation of each protocol; however, any protocol can in theory generate and validate an arbitrary number of version orders.
Hence, by adding a version order reverse to the operation order, any protocol has the possibility to omit write operations.

Based on the findings, we construct a new rule for any protocol: \textbf{non-visible write rule (NWR)}.
The main idea of NWR is to extend existing protocols to \textbf{NWR-extended protocols}, that validate an additional version order.
If the additional version order is validated and the rules of NWR are satisfied, NWR-extended protocols generate omittable write operations with the guarantee of strict serializability and recoverability.
Our approach is extensible for any protocol since we add the version order and its validation, not depending on any baseline implementation.

To illustrate the effectiveness of our approach, we show three NWR-extended protocols: Silo+NWR \cite{Tu2013SpeedyDatabases}, TicToc+NWR \cite{Yu2016Tictoc:Control}, and MVTO+NWR \cite{Reed1978NamingSystem., Bernstein1982ConcurrencySystems}.
All of our extensions are implemented in the same way; we embed information for the additional validation (the additional version order and read/write sets of committed transactions) for each data item as a 128-bits data structure.
On 64-bits machines with double-word atomic operations, an omittable write operation is executed by a single atomic operation.
Compared to the conventional approach with exclusive lockings, our approach promises surprisingly high concurrency and scalability.
In our experiments, Silo+NWR achieves up to 11x faster than the original Silo under write contended workloads.
Meanwhile, every NWR-extended protocol has comparable performance with the original, even in the other workloads such as TPC-C in which there are no omittable write operations.

The rest of this paper is organized as follows: we rethink and formalize the definition of TWR, and discuss the formal aspects of ``omittable'' write operations (Section \ref{sec:formal_aspects_of_iw}).
We introduce our approach of validating the additional version order for generating omittable write operations and propose NWR, the rule for a version order to preserve the correctness (Section \ref{sec:design_overview}).
We show the implementation detail of the proposed approach (Section \ref{sec:implementation}). 
We examine three NWR-extended protocols and confirm that every extended protocol shows a significant performance gain in write-contended workloads (Section \ref{sec:evaluation}). 

\section{Preliminaries}
\label{sec:preliminaries}

\begin{table}[t]
  \small
  \begin{tabular}{|c|l|} \hline
      Symbol & Definition \\ \hline \hline
    $D$ & a set of distinct finite data items ($x, y, z,...$) \\
    $t_i$ & $i$-th transaction. an ordered set of operations \\
    $x_i$ & a version of a data item $x$\\
    $w_i(x_i)$ & a write operation. $t_i$ writes a version $x_i$ \\
    $r_i(x_j)$ & a read operation. $t_i$ reads $x_j$ \\
    $c_i$ & a termination operation which commits $t_i$ \\
    $a_i$ & a termination operation which aborts $t_i$ \\
    $rs_i$ & a set of versions read by $t_i$ \\
    $ws_i$ & a set of versions written by $t_i$ \\
    $\{x\}$ & a set of versions of a data item $x$ \\
    $S$ & a schedule. \footnotemark[1] \\
    $CP(S)$ & returns committed projection of $S$ \\
    $trans(S)$ & returns a set of transactions in $S$ \\
    $\ll$ & a version order for a schedule. \footnotemark[2] \\
    \hline
  \end{tabular}
  \caption{Definitions of main symbols}
  \label{tab:symbols}
\end{table}

\footnotetext[1]{
     A schedule is a pair of operations and $<_s$ where $<_s$ is an order for all pairs of operations from distinct transactions that access the same data item and have at least one write operation \cite{Weikum2001TransactionalRecovery}.
}
\footnotetext[2]{
  A version order for $x$ is any non-reflexive and total ordering of $\{x\}$.
  A version order for a schedule is the union of all version order for data items written by operations in the schedule
  \cite{Weikum2001TransactionalRecovery}.
}

In this paper, we mainly use the notations defined by Weikum and Vossen \cite{Weikum2001TransactionalRecovery}.
Table \ref{tab:symbols} shows definitions of main symbols.
Note that $rs_i$ (read set) and $ws_i$ (write set) are available for all protocols by adding the corresponding sets.
We use symbols and notations of multi-version concurrency control throughout this paper.
Without loss of generality, we can assume that all data items have multiple versions, and all schedules are multi-version schedules \cite{Papadimitriou1986TheControl, Weikum2001TransactionalRecovery, Bernstein1987ConcurrencySystems}.

We consider \textbf{the correctness} of DBMS in the aspects of strict serializability \cite{Herlihy1990Linearizability:Objects} and recoverability \cite{Hadzilacos1988ASystems}.
Strictly serializable protocols always generate schedules equivalent to some serial executions in a single thread on a single machine.
Recoverable protocols ensure that DBMS can always be recovered without any dirty-read \cite{10.1145/223784.223785} anomaly whenever some aborts or crashes happen.
In the following, we explain the details of strict serializability.
The formal definition of recoverability is given in Appendix \ref{sec:proof}. 

\subsection{Serializability and its graph}
\label{sec:def_of_mvsg}
To explain strict serializability, we first introduce serializability.
For a schedule $S$, multiversion view serializability (MVSR) \cite{Papadimitriou1982OnVersions} is defined that $S$ is serializable iff $S$ is \textit{view equivalent} with some serial schedules.
To obtain the view-equivalent serial schedules, Bernstein et al. proposed the multiversion serialization graph (MVSG) \cite{Bernstein1983MultiversionAlgorithms} and prove that $S$ is MVSR iff there exists acyclic MVSG.
An MVSG is generated based on a schedule $S$ and a version order $\ll$.
``$x_i <_v x_j$'' denotes that a version $x_i$ precedes another version $x_j$ in $\ll$.
For given $S$ and $\ll$, an $MVSG(S, \ll)$ has nodes for $trans(CP(S))$ and edges as the followings:
for distinct operations $w_j(x_j)$, $r_i(x_j)$, and $w_k(x_k)$ where $t_i \ne t_k$ and $t_j, t_i, t_k$ in $trans(CP(S))$,
\begin{itemize}
  \setlength\itemsep{0.14em}
\item $t_j \WR{} t_i$
\item If $x_j <_v x_k$  then $t_i \VORW{} t_k$
\item If $x_j >_v x_k$  then $t_k \VOWW{} t_j$
\end{itemize}
We denote $\WR{}$, $\VORW{}$, and $\VOWW{}$ as edges of MVSG and $\to$ as any of edges.
Note that all serializable protocols have their version orders generated by their implementation, such as lockings, timestamps, or transaction ids.
For $MVSG(S, \ll)$, we define a function $RN(t_i)$ which accepts a transaction $t_i$ in $trans(CP(S))$ and returns a set of transactions that includes reachable transactions from $t_i$ on MVSG.

\subsection{Strict serializability}
Strict serializability \cite{Herlihy1990Linearizability:Objects, Fan:2019:OVG:3342263.3360357,10.14778/3342263.3342647} is the property of serializable schedules.
Intuitively, it is the wall-clock ordering constraints among transactions.
In general, distributed concurrency control protocols promise strict serializability for preventing \textit{stale reads}, that read too old versions.
This property is also essential for non-distributed protocols to prevent \textit{stale writes}, that write too old versions.
For instance, a serializable protocol allows all transactions to read and write into an arbitrary version.
A transaction can always read the initial values written by $t_0$ even if there exist some newer versions, and also can always write older versions than the initial values such as $x_j <_v x_0$.
Even if several years may have passed after the initialization of DBMS, not strictly serializable protocols allow such unacceptable results.
To satisfy strict serializability, we must forbid protocols from reordering among operations or versions when these transactions are not concurrent.
A more formal definition of this property is described in Appendix \ref{sec:proof}.

\section{Omittable Write Operations}
\label{sec:formal_aspects_of_iw}
In this section, we first rethink the formal aspects of omittable write operations allowed by TWR.
TWR does not have any formal definitions and propositions for the correctness but explained as follows \cite{Bernstein:2009:PTP:1208930}:

\begin{figure*}[t]
  \centerline{
    \subfloat[$x_1 <_v x_3$. Everyone allows]{
      \includegraphics[clip, width=0.25\textwidth]{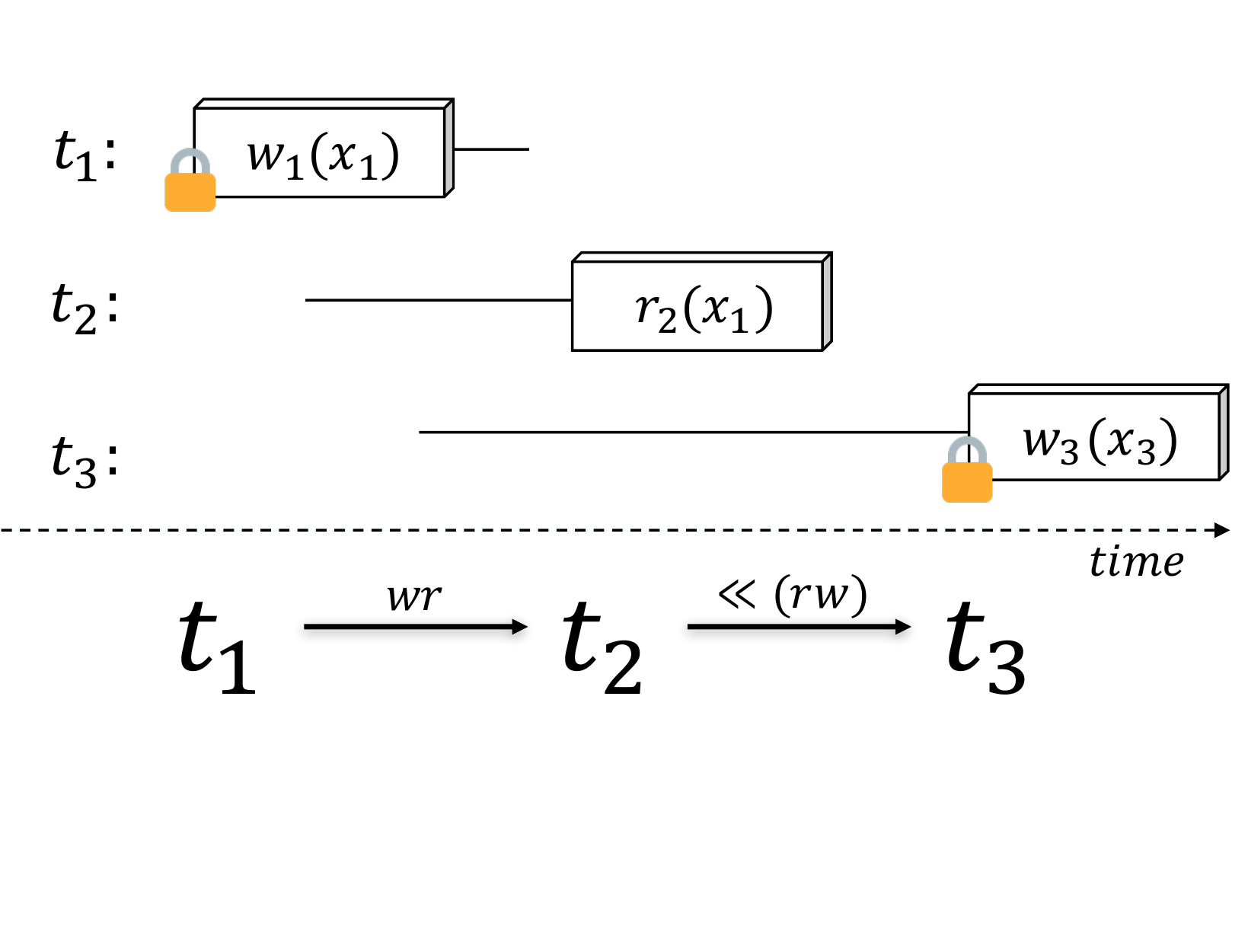}
    }
    \hfill
    \subfloat[$x_1 <_v x_3$. TWR allows]{
      \includegraphics[clip,width=0.25\textwidth]{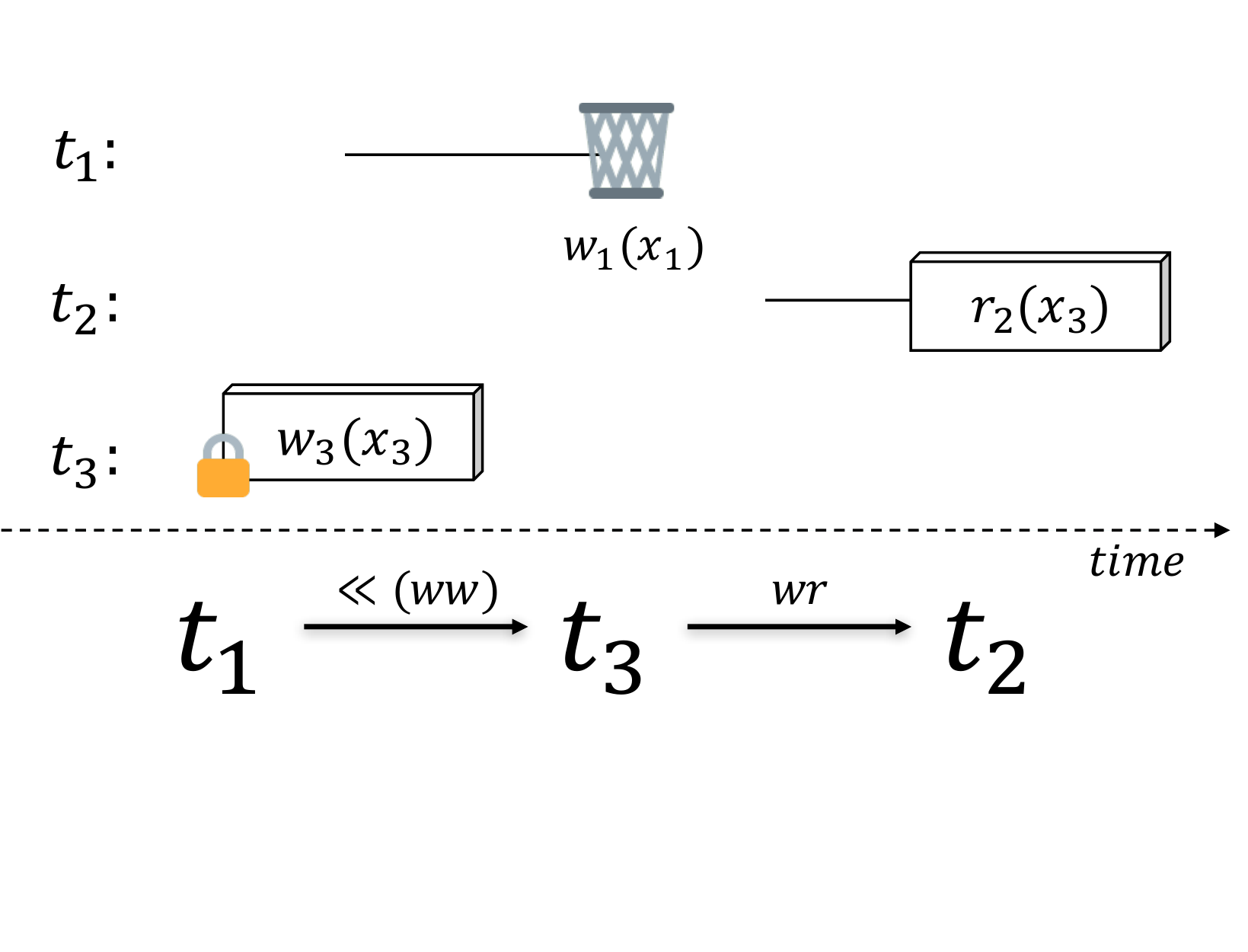}
    }
    \hfill
    \subfloat[$x_1 <_v x_3$. TWR forbids]{
      \includegraphics[clip,width=0.25\textwidth]{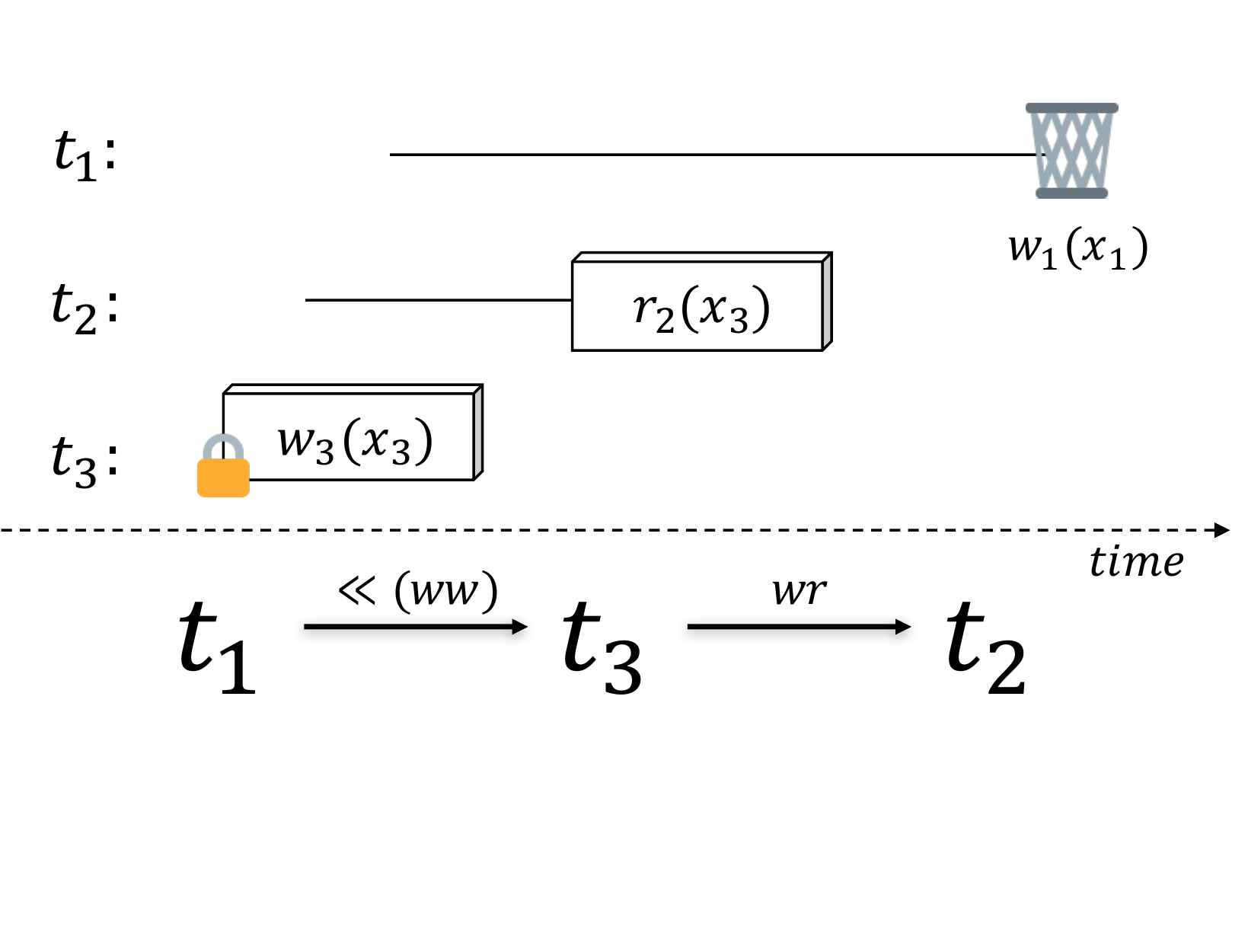}
    }
    \hfill
    \subfloat[$x_3 <_v x_1$. TWR forbids]{
      \includegraphics[clip,width=0.25\textwidth]{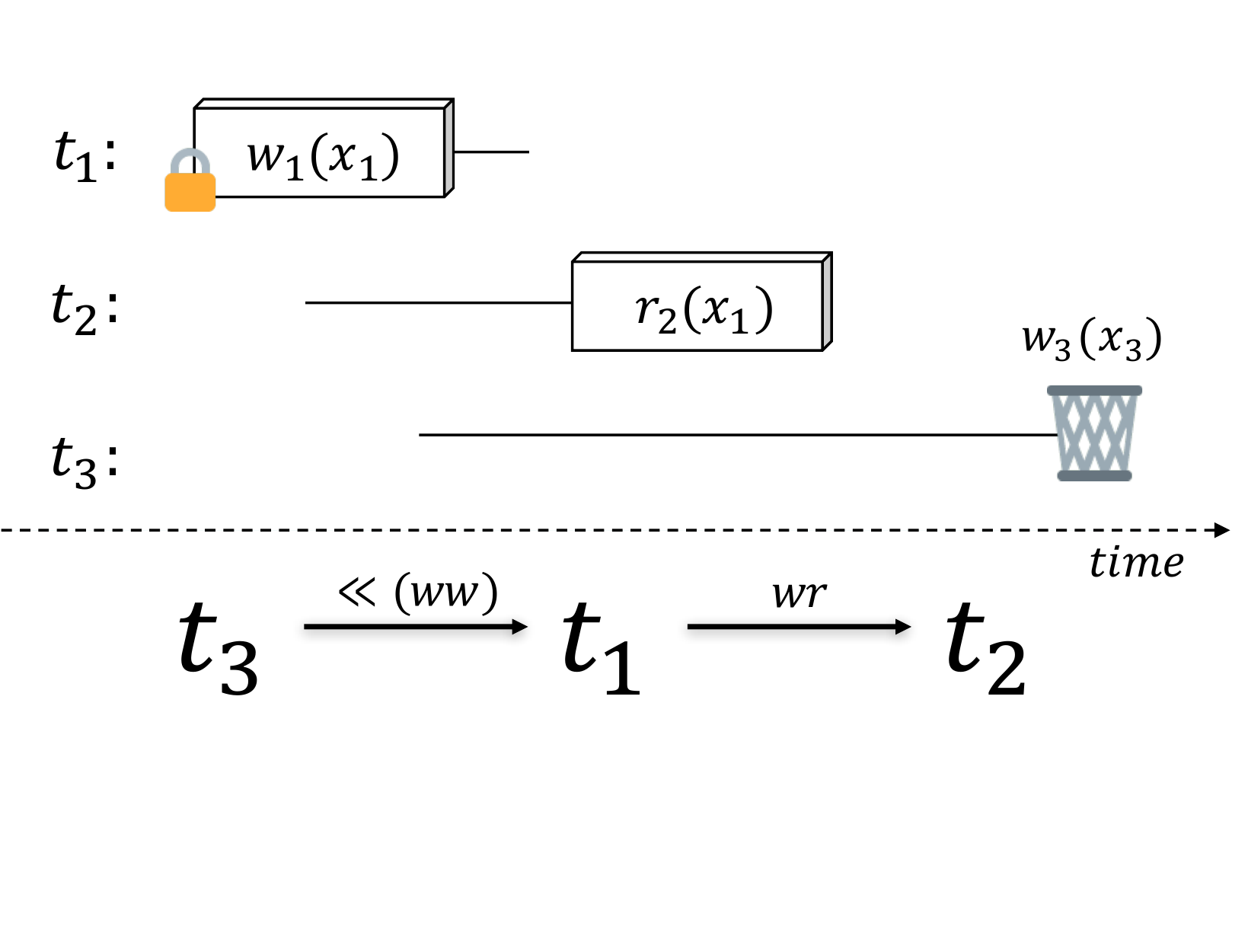}
    }
  }
  \caption{Schedules with version orders. All schedules are strict serializable. Horizontal lines of each transaction indicate the transaction's lifetime. Transactions begin at the start point of their line. (b-d) include an omittable write operation $w_1(x_1)$, $w_1(x_1)$, and $w_3(x_3)$, respectively. Operations painted as trash boxes are omitted. TWR can omit $w_1(x_1)$ in the case of (b), but reject the cases (c) and (d).}
  \label{fig:wrw_optim}
\end{figure*}

\begin{quote}
If an update to data item $x$ arrives at a replica, and the update's timestamp is larger than $x$'s timestamp at the replica, then the update is applied and $x$'s timestamp is replaced by the update's timestamp.
Otherwise, the update is discarded.
...
And eventually, each data item $x$ has the same value at all replicas, because at every replica, the update to $x$ with the largest timestamp is the last one that actually was applied
\end{quote}
where the term ``replica'' means a thread or a machine node.
Here we see that a write operation with the lower timestamp is discarded.
It is based on the fact that the other transactions will never read the discarded write operations; in T/O, a transaction reads the current versions that have the maximal timestamp for each data item.
Therefore, it is unnecessary to apply discarded write operations.
The necessity of applying write operations in TWR derives an essential observation of what is \textit{omittable}:

\begin{observation}
  \label{observation:lazy}
  Execution of a write operation $w_i(x_i)$ can be delayed until a read operation $r_j(x_i)$ arrives.
\end{observation}


Clearly, T/O with TWR utilizes this observation.
Since T/O's read protocol guarantees that discarded versions are never be read, TWR safely postpones executions of write operations indefinitely.
Throughout the paper, we define ``omitting'' as indefinitely delaying execution of an operation.
Based on this observation, we now define an omittable write operation without the concept of timestamps:

\begin{definition}[Omittable]
  \label{def:omittable}
  A write operation $w_j(x_j)$ is \textbf{omittable} if there \textbf{always} does not exist $t_i$ such that $r_i(x_j)$ is in $t_i$.
\end{definition}

Definition \ref{def:omittable} states that any protocol can generate omittable write operations since the timestamps are no longer needed.
For a version $x_j$, if there is the guarantee of non-visibility of $x_j$, we can omit $w_j(x_j)$.
In fact, this guarantee is provided by simply omitting $w_j(x_j)$ without lockings, updating timestamps, or any coordination among the other transactions.
Since omitting does not store any footprints of $x_j$, subsequent transactions cannot know $x_j$'s existence.

Next, we consider the correctness for the generated schedules.
After omitting a write operation, does a strictly serializable protocol even keep this property?
The answer is derived from version orders and concurrency of transactions.
Remind that a schedule is serializable iff there exists a version order which makes the MVSG acyclic, and it is strictly serializable iff there does not exist reordering of operations or versions among non-concurrent transactions.
Hence, if a protocol finds out a version order and validates the concurrency of transactions, the generated schedules are still strictly serializable.
We later show the formal condition and validation for a version order in Section \ref{sec:invisiblewriterule}.
Note that omittable write operations do not affect recoverability; recoverability is violated only when a committed transaction read a version written by a non-committed transaction \cite{Hadzilacos1988ASystems}.

\subsection{Discussion}

We now rethink TWR in terms of version orders and reveal that TWR is the too restrictive rule.
Figure \ref{fig:wrw_optim} depicts four schedules and version orders.
Notice that all schedules are strictly serializable since all MVSGs with illustrated version orders are acyclic, and all transactions are concurrent.
In Figure \ref{fig:wrw_optim}, (b-d) include an omitted operation, but (a) does not.
We assume subscripts represent timestamps of T/O, and thus T/O with TWR allows omitting only in the case of (b).
In (c), T/O must abort $t_1$ since the read-timestamp of $x_3$ has already updated.
In (d), T/O must execute $w_3(x_3)$ since timestamps indicate $t_1 < t_3$; thus $t_3$ must execute $x_3$ as the latest version of $x$ and $x_3$ does not have non-visible property.
It should be pointed out that (a) and (d) show the completely same schedules, and the only difference is the version orders.
In these schedules, we can in theory omit $w_3(x_3)$; however, the existing lock-based or timestamp-based protocols \cite{Bernstein1987ConcurrencySystems,Tu2013SpeedyDatabases,Lim2017Cicada:Transactions, Larson2011High-performanceDatabases} always execute $w_3(x_3)$ as the latest (visible) version of $x$ since their version orders are generated based on the locking order or the timestamp order.
The tight coupling between version order generation and protocol implementation is the current limitation for improving the performance of write operations.

\section{NWR-extended protocols}
\label{sec:design_overview}

To fully utilize the notion of omittability, we have to consider a new approach to generate version orders, not depending on any implementation of existing protocols.
In this section, we propose a novel technique to extend existing protocols by \textit{validating the additional version order $\ll^{NWR}$}.
We extend existing protocols to generate the additional version order $\ll^{NWR}$ and validate its correctness with an additional algorithm.
We call the protocols that have extended based on our technique as \textbf{``NWR-extended protocols''}.
Our approach is based on the principle that any protocol can in theory generate and validate any number of version orders.
In Section \ref{sec:control_flow}, we show the overview of NWR-extended protocols and clarify how to generate the additional version order $\ll^{NWR}$.
In Section \ref{sec:invisiblewriterule}, we define NWR, a formal condition for a version order to preserve the correctness.
In Section \ref{sec:validation_overview}, we show the validation of the $\ll^{NWR}$'s correctness, derived from NWR.

\subsection{Control Flow}
\label{sec:control_flow}

We now explain the overview of how to extend existing protocols.
Figure \ref{fig:flowchart} depicts the control flow of NWR-extended protocols.
We assume a \textit{baseline protocol}, which preserves strict serializability and recoverability, and assume a running transaction $t_j$ is requesting to commit.
Of course when we do not omit any operations in $t_j$, the baseline protocol ensures the correctness; however, NWR-extended protocols interrupt the baseline protocol at $t_j$'s commit request and try to omit all write operations in $t_j$.
We assume that the selected $t_j$ already satisfies recoverability, and a given version order $\ll$ makes MVSG acyclic in the current schedule $S$ such that $c_j$ and $a_j$ are not in $S$.

\textbf{Generate the additional version order.}
At first, NWR-extended protocols generate an additional version order $\ll^{NWR}$.
It should be done at a reasonable cost, but unfortunately, the computational cost of finding all possible version orders is NP-Complete \cite{Bernstein1983MultiversionAlgorithms,Papadimitriou1982OnVersions}.
To reduce the computational cost, we add a component of \textbf{the pivot version}, denoted as $x_{pv}$, for each data item in order to generate an additional version order $\ll^{NWR}$ and validate it.
The pivot versions are the landmark for generating $\ll^{NWR}$;
NWR-extended protocols generate $\ll^{NWR}$ such that $x_j$ is the version just before $x_{pv}$ for all $x_j$ in $ws_j$ and differences between $\ll$ and $\ll^{NWR}$ are always within data items $x_j$ in $ws_j$.
It means that the two graphs $MVSG(S, \ll)$ and $MVSG(S, \ll^{NWR})$ are isomorphic.
It is unnecessary to generate and validate the whole MVSG for $\ll^{NWR}$; only a connected subgraph, which includes $t_j$, is necessary for validations.

\textbf{Validations of the correctness.}
We next add two components for validation of the correctness.
These components validate that the pair of $S \cup \{c_j\}$ and $\ll^{NWR}$ does not violate serializability and strict serializability.
For serializability, we check the acyclicity of $MVSG(S \cup \{c_j\}, \ll^{NWR})$.
For strict serializability, we check $t_j$ and all $t_{pv}$, that wrote $x_{pv}$, are concurrent.
Note that strict serializability allows reading or writing old versions only among concurrent transactions, and subsequent non-concurrent transactions must read the latest versions for all data items.
Hence, all $x_j$ in $ws_j$ must be non-latest versions to omit them with the guarantee of non-visibility.
The version order $\ll^{NWR}$ based on the pivot versions always satisfies this condition.

\begin{figure}[t]
    \includegraphics[width=0.45\textwidth]{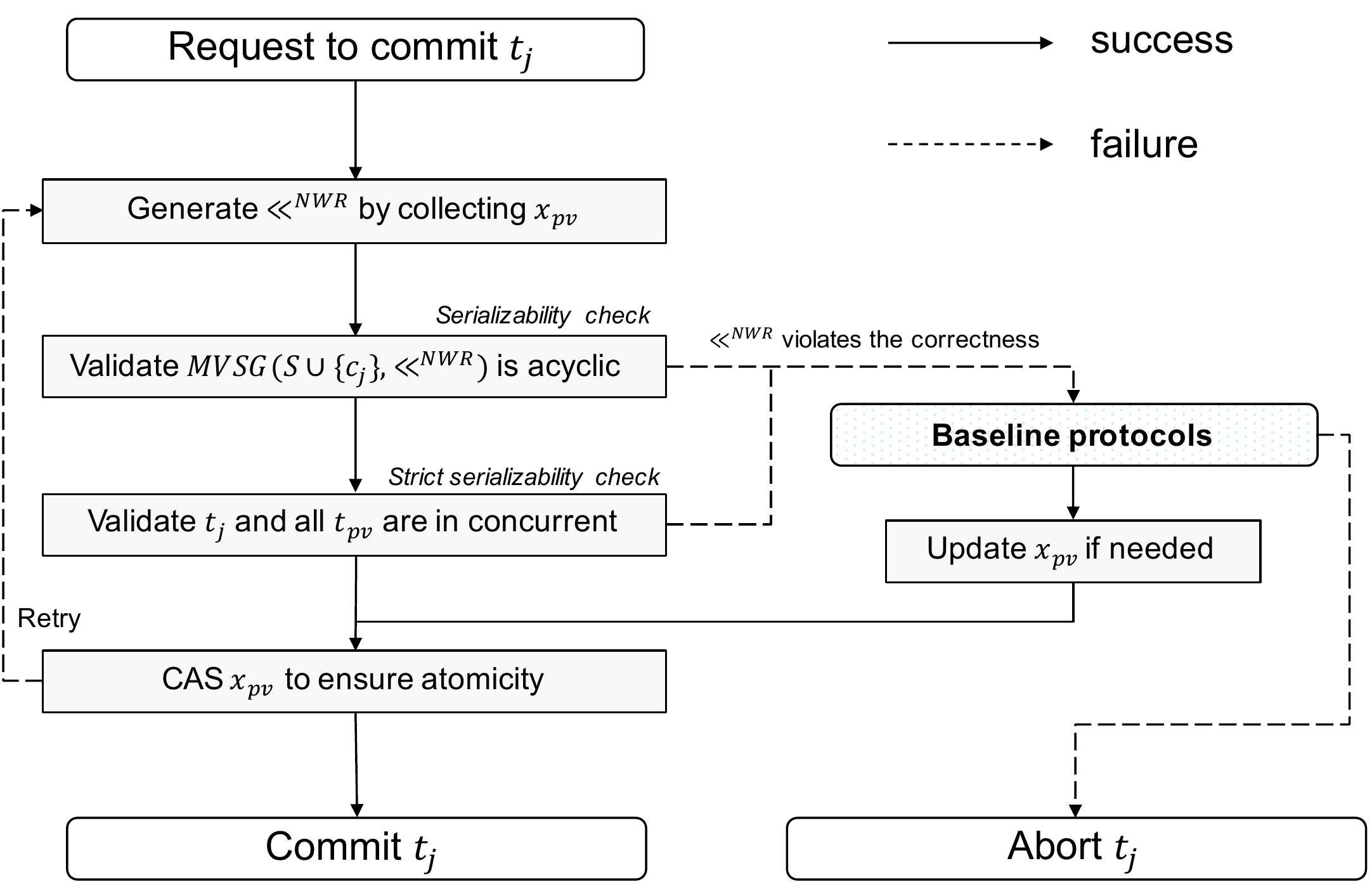}
   \caption{The control flow of NWR-extended protocols.} 
    \label{fig:flowchart}
\end{figure}

\textbf{CAS to ensure atomicity.}
After validations, NWR-extended protocols execute compare and swap (CAS) \cite{Guide2011IntelManual,Technology2012AMD64Date} operations to ensure atomicity of validations.
When all validations and CAS operations pass, then the generated schedule with $\ll^{NWR}$ preserves the correctness; and thus, NWR-extended protocols go to omit all write operations of $t_j$ and commit $t_j$ without any exclusive locking, buffer update, index update, and logging.
In Section \ref{sec:mergedsets_update} we describe the more details of lock-free updates.

\textbf{No false aborts.}
For the case of validation failures, NWR-extended protocols transfer the control to baseline protocols.
Then baseline protocols validate serializability and strict serializability again with their version order, and they determine $t_j$'s termination operation.
Notably, all added components in the control flow do not abort $t_j$ directly.
Only baseline protocols can abort $t_j$.
Our technique offers some overhead for added components but does not increase the abort ratio of transactions.

\textbf{Updating the pivot versions.}
The validations of strict serializability tell us an important observation of how to update the pivot versions.
If we choose a not suitable version as a pivot version $x_{pv}$, validations always fail.
From the definition of MVSG, if $x_{pv}$ is generated by a ``read-modify-write'' operations, the validation of serializability always fails.
More precisely for a running transaction $t_j$, if there exists read-modify-write $r_k(x_i)w_k(x_k)$ and $\ll^{NWR}$ includes $x_i <_v x_j <_v x_k$, MVSG always has the cyclic path $t_k \VORW t_j \VOWW{} t_k$.
Therefore, for serializability, the pivot versions should be updated by \textit{blind-writes} \cite{Weikum2001TransactionalRecovery}.
Furthermore, strict serializability indicates another restriction.
If $\ll^{NWR}$ includes $x_j <_v x_{pv}$ and $t_{pv}$ is not concurrent with $t_j$, strict serializability is violated and the validation fails.
Now we obtain the next proposition:
\begin{proposition}
  A write operation is omittable with the correctness only if there exists another concurrent \textit{blind-write} operation.
\end{proposition}
From this we update the pivot versions in the case a running transaction $t_j$ execute $x_j$ such that $w_j(x_j)$ is blind-write.
This observation also tells us that inserting operations can update the pivot versions, but they are not omittable with the correctness.
For an inserting operation $w_j(x_j)$, the other blind-write is only $w_0(x_0)$.
Clearly $t_0$ is not concurrent with $t_j$ and thus the control flow fails to omit $w_j(x_j)$.

\subsection{Non-visible Write Rule}
\label{sec:invisiblewriterule}

We next define \textbf{Non-visible Write Rule (NWR)}, the formal condition for the additional version order, to clarify what kind of version order preserves the correctness and how to validate strict serializability in the control flow.
NWR consists of five rules: non-visible property rule (NV-Rule), preserving version order rule (PV-Rule), recoverability rule (RC-Rule), serializability rule (SR-Rule), and strict serializability rule (ST-Rule).
NV-Rule forces that all $t_j$'s write operations are non-latest versions, and they are omittable.
PV-Rule forces that the difference between $\ll$ and $\ll^{NWR}$ must be only related to data items $x_j$ in $ws_j$.
As has been noticed above, this rule efficiently reduces the computational cost.
Remaining the three rules are derived straightforwardly from the definition of serializability \cite{Bernstein1983MultiversionAlgorithms}, strict serializability \cite{Herlihy1990Linearizability:Objects}, and recoverability \cite{Hadzilacos1988ASystems}, respectively.

\begin{definition}[Non-visible Write Rule]
  \label{def:invisiblewriterule}
  Let $S, \ll, t_j$ be a schedule, a version order for $S$, and a running transaction such that $t_j \in trans(S) \land c_j, a_j \notin S$, respectively.
  Let $S$ be recoverable, $MVSG(S, \ll)$ be acyclic, and $S$ be strictly serializable.
  For a version order $\ll'$, we define NWR as follows:

  \begin{description}
    \item [ \ \ NV-Rule] \mbox{}\\
           $\forall x_j \exists x_k (x_j \in ws_j \Rightarrow  x_j <_v' x_k)$ where $<_v'$ in $\ll'$

    \item[  \ \ PV-Rule] \mbox{}\\
          $\forall x_i \forall x_k ((x_i,x_k \in \{x\} \land x_i \ne x_k \ne x_j \land x_i <_v x_k) \Rightarrow x_i <_v' x_k)$ where $<_v$ in $\ll$ and $<_v'$ in $\ll'$

    \item[  \ \ SR-Rule] \mbox{}\\
          $t_j \notin RN(t_j)$

    \item[  \ \ ST-Rule]\mbox{}\\
          $\forall t_k \exists p_j ((t_k \in RN(t_j) \land p_j \in t_j) \Rightarrow  p_j <_S c_k)$ where $<_S$ in $S$

    \item[  \ \ RC-Rule] \mbox{}\\
          $\forall x_i(x_i \in rs_j \Rightarrow c_i \in S)$

  \end{description}
  where $RN$ is defined on $MVSG(S \cup \{c_j\}, \ll')$.
\end{definition}

The following theorem states that if NWR is satisfied, a protocol can commit $t_j$ with the correctness:

\begin{theorem}[Correctness]
  \label{theo:NWR_is_correct}
  If a version order $\ll'$ satisfies NWR,
  $S \cup \{c_j\}$ is strictly serializable and recoverable.
\end{theorem}

The proof is given in Appendix \ref{sec:proof}.

Hence, if a version order $\ll^{NWR}$ generated by NWR-extended protocols satisfies NWR, a running transaction $t_j$ can in theory omit all write operations with the correctness.

\subsection{Validation Overview}
\label{sec:validation_overview}
At the time a running transaction $t_j$ is selected and $\ll^{NWR}$ is generated based on the pivot versions, it is clear that the rules of NWR are satisfied except SR-Rule and ST-Rule.
We next explain the overview of validations for these two rules.
A key observation is that both rules collect and validate reachable transactions $RN(t_j)$.
We now separates directly reachable transactions into two sets:
  $$overwriters_j := \{t_k| t_j \VORW{} t_k \in MVSG(S, \ll^{NWR})\}$$
  $$successors_j := \{t_k| t_j \VOWW{} t_k \in MVSG(S, \ll^{NWR})\}$$
The following theorem is then immediate:

\begin{theorem}[Reachable transactions]
  \label{theo:what_are_RN}
  If RC-Rule is satisfied, then any path in $RN(t_j)$ is starting from $t_j \to t_k$ such that $t_k$ is in $overwriters_j$ or $successors_j$.
\end{theorem}
\begin{proof} (sketch)
As we have been noted in Section \ref{sec:def_of_mvsg}, edges of MVSG are distinguished the three types: $\WR{}, \VORW{}$, and $\VOWW{}$. 
Since RC-Rule is satisfied and $S$ is recoverable, there does not exist any edge $\WR{}$ starting from $t_j$; no committed transactions read any version $x_j$ in $ws_j$.
Therefore, the type of an edge $t_j \to t_k$ for any $t_k$ is always either $\VORW{}$ or $\VOWW{}$.
\end{proof}

So we can separate the validation of NWR-extended protocols into two sub-validations related to $successors_j$ and $overwriters_j$.
We now summarize the additional validations of NWR-extended protocols as follows: for each $t_k$ in $overwriters_j$ or $successors_j$, (1) checking $t_k$ is not reachable into $t_j$ (SR-Rule), and (2) checking $t_k$ is concurrent with $t_j$ (ST-Rule).

\subsubsection{Validation of overwriters}
We next show that it is not necessary to invent validations for $overwriters_j$.
Since elements of $overwriters_j$ depends on $t_j$'s read operations, and our approach focuses only on write operations, transactions of $overwriters_j$ are never changed from baseline protocols.
Therefore, there exist reusable validation algorithms.
We next show how to validate $overwriters_j$ for each SR-Rule and ST-Rule.

\textbf{SR-Rule validation of overwriters.}
If some transaction $t_k$ in $overwriters_j$ violates SR-Rule, there exists some path $t_j \VORW{} t_k \to  ... \to t_j$.
That is, some version $x_i$ in $rs_j$ is overwritten by $t_k$.
This type of edge is well-known as the \textit{anti-dependency} \cite{2000:GIL:846219.847380, Papadimitriou1986TheControl, Hadzilacos1988ASystems} and all serializable protocols can detect or prevent the cycles starting from an anti-dependency edge by using lockings \cite{Gray1992TransactionTechniques}, timestamps \cite{Reed1978NamingSystem.}, or validation phases \cite{Kung1981OnControl}. 
Thus, we can reuse implementations of baseline protocols straightforwardly.

\textbf{ST-Rule validation of overwriters.}
If some transaction $t_k$ in $overwriters_j$ violates ST-Rule, there exists some too old version $x_i$ in $rs_j$; however,  $t_j$ cannot read such stale versions since baseline protocols ensure strict serializability.
Thus, we do not have to consider ST-Rule validation of $overwriters_j$.

\subsubsection{Validation of successors}
\label{sec:validation_overview}
We have seen that it is unnecessary to validate $overwriters_j$; however, we have to add new validations for $successors_j$.
Algorithm \ref{alg:successors_validation_pseudo} provides the pseudo-code of the validations of both SR-Rule and ST-Rule for $successors_j$.
In (A) of Algorithm \ref{alg:successors_validation_pseudo}, we collect every transaction $t_m$ such that $t_m$ is in $successors_j$ or $t_m$ is reachable from $t_k$ in $successors_j$.
In (B), we validate ST-Rule by checking the concurrency of all $t_m$ and $t_j$.
Next, (C) and (D) are validations of SR-Rule.
We focus and find the last edge $t_m \to t_j$ in cyclic paths.
Note that any transaction $t_m$ does not have a path $t_m \VOWW{} t_j$ since $S$ is recoverable; the edge $\VOWW{}$ is added into MVSG only if some committed transactions read some $x_j$ in $ws_j$ (it is impossible because of RC-Rule).
Therefore, we need only to check the two types of the last edge $t_m \WR{} t_j$ and $t_m \VORW{} t_j$; (C) checks the former and (D) checks the latter.
(C) validates that some $y_n$ in $rs_j$ is newer than $y_m$ in $ws_m$ for some $t_m$.
If it holds, there exists a cyclic path $t_j \VOWW{} ... \to t_m \WR{} t_j$.
Likewise, (D) validates that some $y_j$ in $ws_j$ is newer than $y_g$ in $rs_m$ for some $t_m$.
If it holds, there exists a cyclic path $t_j \VOWW{} ... \to t_m \VORW{} t_j$.
Consequently, when we pass validations (C) and (D), any $t_k$ in $successors_j$ is not reachable into $t_j$.

\begin{algorithm}[t]
  \begin{algorithmic}
    \Require{$t_j, \ll^{NWR}$}
    \State{$T = \{t_k, t_i | t_k \in successors_j \land t_i \in RN(t_k) \}$}
    \Comment{(A)} 
    \State 

    \ForAll{$t_m$ in $T$}

    \If{$t_m$ commits before $t_j$'s beginning}
      \Comment{(B)} 
      \State 
      \Return{ST-Rule is not satisfied}
    \EndIf
    
    \State 

    \ForAll{$y_m$ in $ws_m$}
    \ForAll{$y_n$ in $rs_j$}
    \If{$y_m <_v y_n$ or $m = n$}
    \Comment{(C) $\WR{}$ to $t_j$} 
    \State 
    \Return{MVSG is not acyclic}
    \EndIf
    \EndFor
    \EndFor

    \State 
    \ForAll{$y_g$ in $rs_m$}
    \ForAll{$y_j$ in $ws_j$}
    \If{$y_g <_v y_j$}
    \Comment{(D) $\VORW{}$ to $t_j$} 
    \State 
    \Return{MVSG is not acyclic}
    \EndIf
    \EndFor
    \EndFor
    \EndFor
    \State{} 

    \Return{MVSG is acyclic}
  \end{algorithmic}
  \caption{validations for $successors_j$ where $<_v$ in $\ll^{NWR}$}
  \label{alg:successors_validation_pseudo}
\end{algorithm}


\section{Lock-free Implementation}
\label{sec:implementation}

We can now validate SR-Rule and ST-Rule with the validation for $successors_j$.
However, the naive implementation of Algorithm \ref{alg:successors_validation_pseudo} has huge overheads since it forces all transactions to store their footprints into shared memory; there is an assumption that we have both $rs_k$ and $ws_k$ for all reachable transactions from $t_k$ in $successors_j$, at the commit request from $t_j$.
To solve this problem, we propose \textbf{the pivot version objects}, which provide the efficient lock-free implementation of NWR-extended protocols.
We embed the pivot version object for each data item $x$.
Each object stores each pivot version $x_{pv}$ and footprints of concurrent transactions that are reachable from $t_{pv}$ on MVSG.
These objects play crucial roles of the control flow of NWR-extended protocols, described in Section \ref{sec:control_flow}.
Briefly, for a transaction $t_j$, we find $x_{pv}$ for all $x_j$ in $ws_j$, generate $\ll^{NWR}$, and validate $successors_j$ by using the pivot version objects.
In Section \ref{sec:the_data_structure}, we explain the layout of the pivot version objects.
In Section \ref{sec:mergedsets_update}, we show the detailed and efficient control flow of NWR-extended protocols in lock-free fashion.

\subsection{Layout of the pivot version objects}
\label{sec:the_data_structure}

Each pivot version object includes four fields: an epoch number, the pivot version (PV), a merged read set (mRS), and a merged write set (mWS).
Figure \ref{fig:datastructure} shows the layout of the pivot version object.

\begin{itemize}

  \item
        \textbf{Epoch.}
        This field holds 32-bits an epoch number, which enables to validate ST-Rule.
        We use \textit{epoch-based group commit} \cite{Tu2013SpeedyDatabases, Chandramouli2018FASTER:Updates} for all NWR-extended protocols.
        Epoch-based group commit divides the wall-clock time into \textit{epochs}, and assigns each transaction into the corresponding epoch.
        Transactions in an epoch are grouped, and their commit operations are delayed until the epoch is \textit{stable}.
        This nature of delaying commit operations is favorable for ST-Rule;
        it enables us to check that two transactions are concurrent in a simple way.
        Since transactions in the same epoch are committed at the same time, we can say that transactions in the same epoch are concurrent.
        The epoch field stores the epoch number of the transaction $t_{pv}$, which write the pivot version of this data item.
        If the epoch number does not match the current ($t_j$'s) epoch, ST-Rule does not satisfied.
        We process the epoch number as the same with Silo \cite{Tu2013SpeedyDatabases}; after all transactions in an epoch are committed, subsequent transactions are assigned in the next epoch.
  \item
        \textbf{The pivot version (PV).}
        This field stores the \textit{version number} of pivot version $x_{pv}$, which enables to generate the additional version order $\ll^{NWR}$.
        Version numbers are decided in the version order $\ll$ given by baseline protocols and reset for each epoch.
        Generally, baseline protocols have per-data item version numbers such as write-timestamp \cite{Yu2016Tictoc:Control,Lim2017Cicada:Transactions} or transaction id \cite{Tu2013SpeedyDatabases}, thus we can reuse them.
        For a running transaction $t_j$, a version number of $x_j$ in $ws_j$ can take two possible values:
        for the additional version order $\ll^{NWR}$, the version number of $x_j$ is given by decrement of $x_{pv}$ since the version $x_j$ is just before version of $x_{pv}$ in $\ll^{NWR}$.
        If $t_j$ does not commit with $\ll^{NWR}$ but commit with a version order $\ll'$ generated by baseline protocols, the version number of $x_j$ must reflect $\ll'$.
        We update this field at the commit of transactions that perform the first \textit{blind-update} for each epoch.
  \item
        \textbf{MergedRS (mRS) and MergedWS (mWS).}
        These fields store mRS and mWS, 32-bits hash table that enable to validate SR-Rule for $successors_j$ described in Section \ref{sec:validation_overview}.
        Figure \ref{fig:merged_sets} illustrates the layouts.
        Each data item $x, y, z, ..., \in D$ is mapped into the corresponding slot by a hash function $h$.
        Each hash table consists of eight 4-bits slots, and stores an union of $rs_k$ or $ws_k$ for all transactions $\{t_k\}$ such that $t_k$ is in $RN(t_{pv})$ or $t_k = t_{pv}$.
        Each $t_k$ stores the version number for each $x_i$ in $rs_k$ and $x_k$ in $ws_k$ into the corresponding slot if the version number is less than the current value, or the current value is zero.
        Each slot holds only the lowest version number as 4-bits integer.
        When a version number is larger than the upper bound of 4-bits integers, we truncate it to $(2^4-1)$.
        These fields reset when the epoch number proceeds.
\end{itemize}

We encode each field as 32-bits, and thus every data item maintains 128-bits the pivot version object.
On the 64-bits machine with 128-bits atomic operations such as CMPXCHG16B \cite{Guide2011IntelManual,Technology2012AMD64Date}, we can handle the pivot version objects with lock-free fashion.

\begin{figure}[t]
  \includegraphics[width=0.45\textwidth]{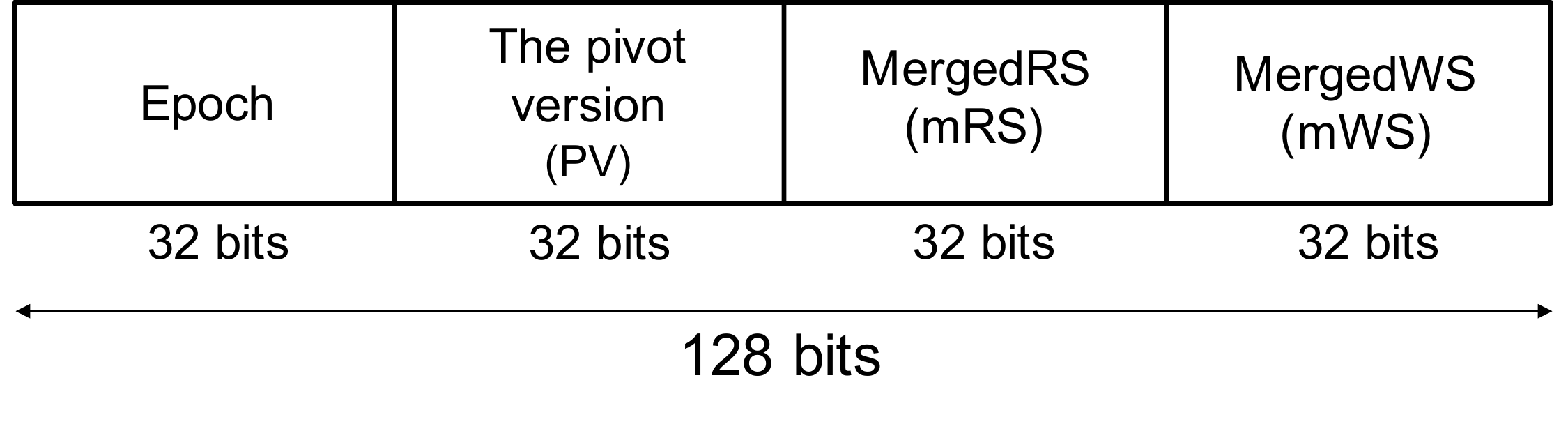}
  \caption{Layout of a pivot version object. It consists of four 32-bits fields.}
  \label{fig:datastructure}
\end{figure}

\begin{figure}[t]
  \includegraphics[width=0.45\textwidth]{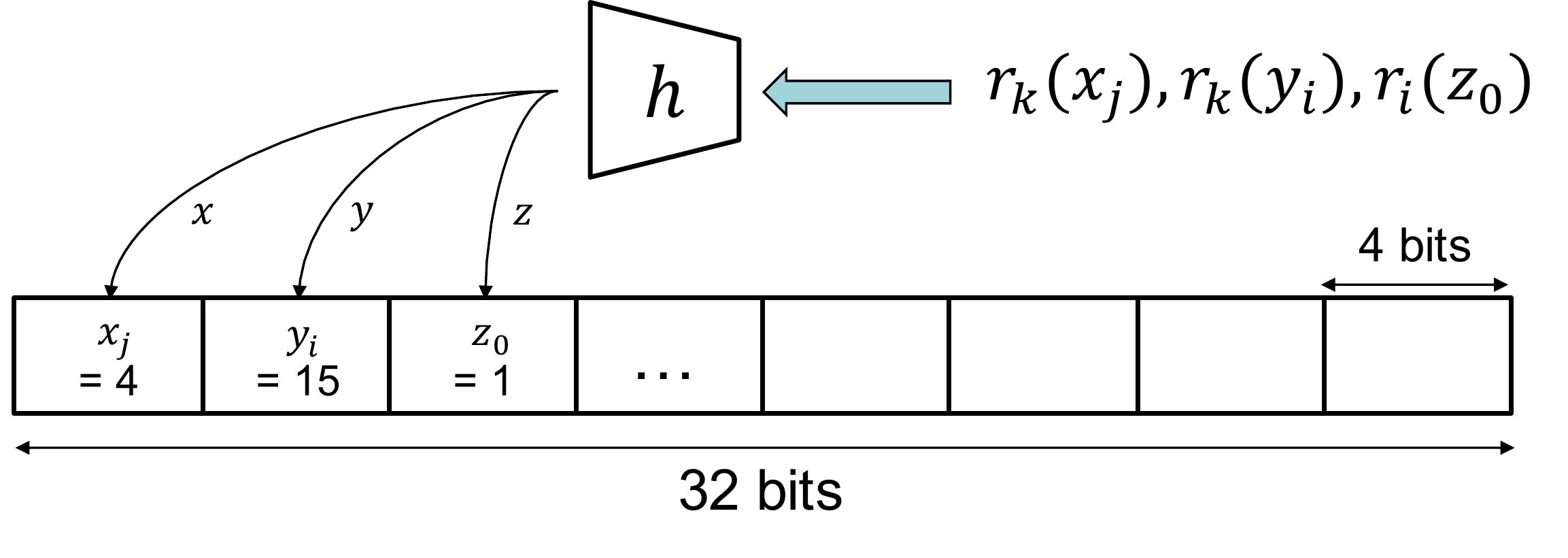}
  \caption{mRS and mWS are 32-bits hash table consisting of eight 4-bits version numbers. The version number reflects the given version order $\ll$.}
  \label{fig:merged_sets}
\end{figure}

\textbf{Validation with the pivot version objects.}
Algorithm \ref{alg:successors_validation} shows our detailed implementation of the pseudo-code shown in Algorithm \ref{alg:successors_validation_pseudo}. 
In (1), protocols fetch the pivot version objects for all $x_j$ in $ws_j$.
(2) is the validation of ST-Rule by comparing epoch numbers.
(3) and (4) are similar to (C) and (D) in Algorithm \ref{alg:successors_validation_pseudo} but present two cases of false positives.
First, hash collisions incur false positives.
When a hash function $h$ assigns data items $y$ and $z$ into the same slot, we may compare version numbers for different data items $x$ and $y$ and it incurs unnecessary aborts.
The other case of false positives is caused by 4-bits integer truncation of mRS and mWS.
When a version number in $rs_j$ or $ws_j$ is larger than the upper bound of 4-bits integers, we cannot compare the version number correctly.
Then we conclude that MVSG might not be acyclic, and thus NWR-extended protocols do not commit $t_j$ with $\ll^{NWR}$.

\begin{algorithm}[t]
  \caption{validations of $successors_j$ with false positives, where $<_v \in \ll^{NWR}$}  \label{alg:successors_validation}
  \begin{algorithmic}
    \Require{$t_j$}

    \ForAll{$x_j$ in $ws_j$}
    \State{s := get\_the\_pivot\_version\_objects($x$)}
    \Comment{(1)}

    \State 
    \If{s.epoch $\ne$ $t_j$'s epoch}
    \State 
    \Return{ST-Rule does not satisfied}
    \Comment{(2)} 
    \EndIf

    \State 
    \ForAll{$y_m$ in s.mWS}
    \If{$z_n \in rs_j$ and $h(y) = h(z)$}
    \If{$y_m \leq z_n$ or $(2^4-1) \leq z_n$}
    \Comment{(3)} 
    \State 
    \Return{MVSG may not be acyclic}
    \EndIf
    \EndIf
    \EndFor

    \State 
    \ForAll{$y_g$ in s.mRS}
    \If{$z_j$ $\in ws_j$ and $h(y) = h(z)$}
    \If{$y_g < z_j$ or $(2^4-1) \leq z_j$}
    \Comment{(4)} 
    \State 
    \Return{MVSG may not be acyclic}
    \EndIf
    \EndIf
    \EndFor
    \EndFor
    \State{} 

    \Return{MVSG is acyclic}
  \end{algorithmic}
\end{algorithm}

\subsection{Detailed control flow}
\label{sec:mergedsets_update}

Algorithm \ref{alg:NWR_updating_ds} shows the details of the control flow.
In (1), we update mRS for all $rs_j$.
In (2), we check SR-Rule and ST-Rule of both $overwriters_j$ and $successors_j$.
When some validation fails, NWR-extended protocols quit to validate $\ll^{NWR}$ and then delegate the control flow to baseline protocols to determine $t_j$'s termination operation.
(3) indicates the case when $t_j$ performs the first blind-write $w_j(x_j)$ in the current epoch.
In this case, NWR-extended protocols update all four fields of the pivot version object of $x$ to $x_{pv} = x_j$.
We reset the version number, mRS, and mWS at this time; we update $x_{pv} = 1$, the first pivot version of this epoch.
In (4), $t_j$ must update mRS and mWS since at this time $x_{pv}$ is replaced with $t_j$ and thus subsequent transactions need to validate SR-Rule with $t_j$.
It is unnecessary to update both the epoch and the pivot version fields.

Throughout the control flow, we must access the pivot version objects with the following invariants: 1) both epoch and PV field keeps the information of $t_{pv}$ and must be updated together, and 2) mRS and mWS store merged footprints for all $\{t_{pv}\} \cup RN(t_{pv})$.
When some fields hold inconsistent information, validations may be passed illegally.
The pivot version objects' 128-bits data layout enables us to preserve these invariants with lock-free fashion.
At the beginning of the commit protocol, NWR-extended protocols copy all pivot version objects to another location.
After validation and modification, we call atomic operations such as CAS or DCAS to update all fields atomically.
When CAS or DCAS has failed, we retry the commit protocol from the beginning.

\begin{algorithm}[t]
  \caption{Control flow of NWR-extended protocols}
  \label{alg:NWR_updating_ds}
  \begin{algorithmic}
    \Require{$t_j$}

    \State{\textit{\# Begin atomic section}}
    \ForAll{$x_i$ in $rs_j$}
    \State{s := get\_the\_pivot\_version\_objects(x)}
    \If{s.epoch is the current epoch}
    \ForAll{$y_k$ in mRS}
    \If{$h(x) = h(y)$ and $x_i < y_k$}
    \State{mRS[$h(x)$] := $x_i$}
    \Comment{(1)} 
    \EndIf
    \EndFor
    \EndIf
    \EndFor
    \State

    \ForAll{$x_j$ in $ws_j$}
    \State{s := get\_the\_pivot\_version\_objects(x)}
    \State 

    \If{validation of $overwriters_j$ or $successors_j$ fail}
      \If{baseline protocols determine $a_j$}
      \Comment{(2)} 
        \State{}
        \Return{abort}
      \EndIf
    \EndIf

    \If{s.epoch is the current epoch}
    \If{$w_j(x_j)$ is the first blind-write in s.epoch}
    \State{s.PV := $x_j$}
    \Comment{(3)} 
    \State{s.epoch := $t_j$.epoch}
    \State{init s.mRS with $rs_j$}
    \State{init s.mWS with $ws_j$}
    \Else
    \Comment{(4)} 
    \ForAll{$y_k$ in s.mRS}
    \If{$z_i \in rs_j$ and $h(y) = h(z)$}
    \If{$z_i < y_k$}
    \State{s.mRS[$h(y)$] := $z_i$}
    \EndIf
    \EndIf
    \EndFor

    \ForAll{$y_k$ in s.mWS}
    \If{$z_j \in ws_j$ and $h(y) = h(z)$}
    \If{$z_j < y_k$}
    \State{s.mWS[$h(y)$] := $z_j$}
    \EndIf
    \EndIf
    \EndFor
    \EndIf
    \EndIf
    \EndFor
    \State{\textit{\# End atomic section}}
    \State 
    \State{}
    \Return{commit}
    \Comment{all $w_j(x_j)$ in $t_j$ are omittable} 

  \end{algorithmic}
\end{algorithm}

\section{Evaluation}
\label{sec:evaluation}
We now show our evaluation of three NWR-extended protocols: Silo+NWR \cite{Tu2013SpeedyDatabases}, TicToc+NWR \cite{Yu2016Tictoc:Control}, and MVTO+NWR \cite{Reed1978NamingSystem.,Bernstein1982ConcurrencySystems, Lim2017Cicada:Transactions}.
Our focus of the evaluation is to clarify the followings:

\begin{itemize}
  \setlength\itemsep{0.15em}
\item{Effectiveness: omittable write operations improve the performance in write-contended workloads.}
\item{Low-overhead property: Every NWR-extended protocol keeps comparable performance to the original, even in read-intensive workloads.}
\item{Limitation: our approach is efficient only in the case when the workload contains blind-writes.}
\end{itemize}

The performance results demonstrate that, if workloads contain blind-writes, our approach is useful for both read-intensive and write-intensive workloads.

\begin{figure}[t]
  \includegraphics[width=0.45\textwidth]{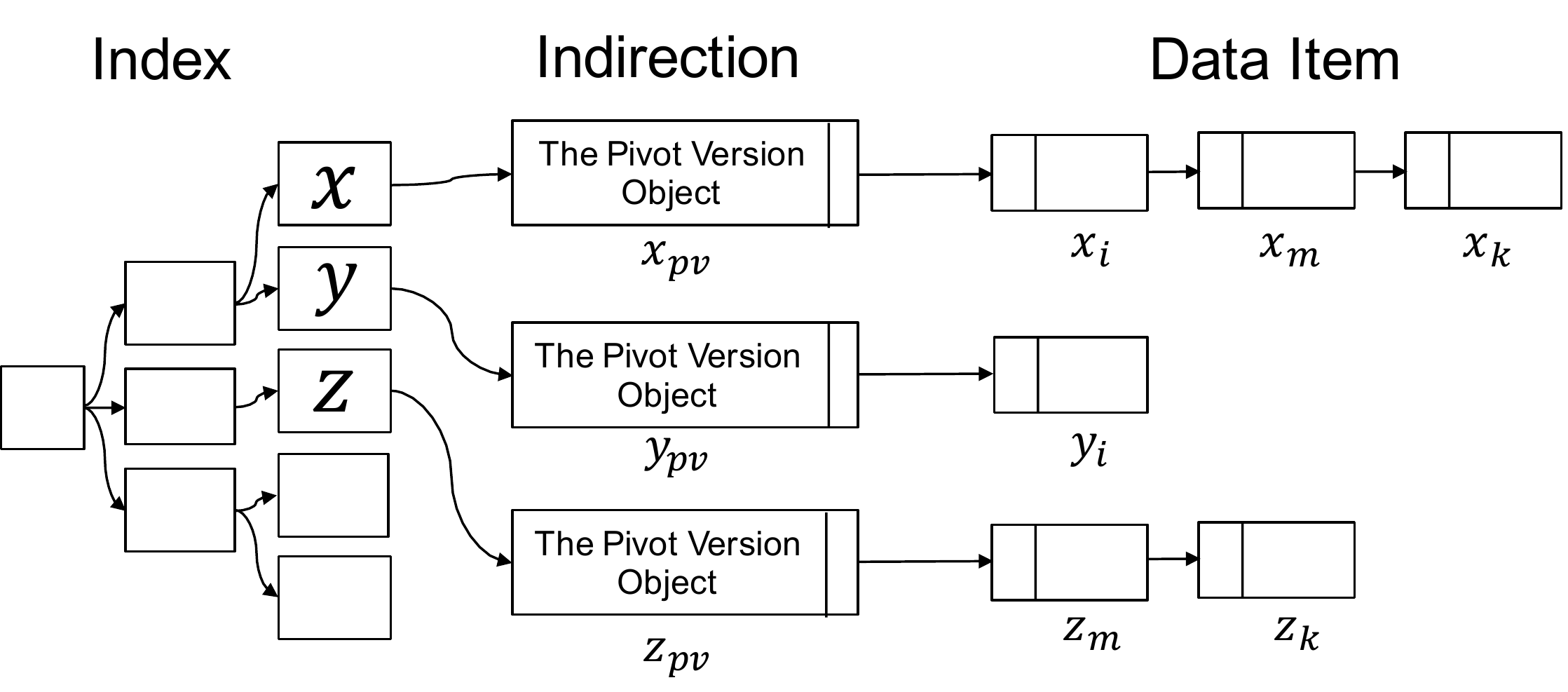}
  \caption{Overall structure of our prototype}
  \label{fig:whatisfv}
\end{figure}

\subsection{Experimental Setup}

For experiments, we implement our prototype: a lightweight, non-distributed, and embedded transactional key-value storage written in C++.
Our prototype consists of in-memory storage, various CC algorithms, a tree-based index forked by Masstree \cite{Mao2012CacheStorage}, and a parallel-logging manager \cite{Zheng2014FastParallelism, Johnson2010Aether:Logging}.
Figure \ref{fig:whatisfv} shows the overall structure of our prototype.
Every leaf-node of the index has the pointer to an indirection.
Each indirection includes a) 64-bits pointer to the data item and b) the pivot version object described in Section \ref{sec:implementation}.
Each data item is represented as a linked-list ordered by versions as newest-to-oldest.
For all single-version protocols, every linked-list always consists of a single node.

We deploy all experiments on a 144-core machine with four Intel Xeon E7-8870 CPUs and 1TB of DRAM.
Each CPU socket has 18 physical cores and 36 logical cores by hyperthreading.
The results of experiments in over 72 cores may show sub-linear scaling due to contention within physical cores.
Each socket has 45MB L3 shared cache.
Each worker thread allocates the memory from separated space specified with the Linux numactl.
We compile all queries at build time, and thus our experiments do not use both networked clients and SQL interpreters.
Each worker thread has the thread-local workload generator to input transactions by itself.

\subsubsection{Extension details}
In our experiments, we extend the following three state-of-the-art protocols:

\textbf{Silo.}
This optimistic protocol adopts the backward validation \cite{Kung1981OnControl}, and the epoch-based group commit.
Silo generates a version order by exclusive lockings.
When a transaction $t_j$ requests to commit, Silo acquires exclusive lockings for all $x_j$ in $ws_j$.
Locking ensures that all $x_j$ are the latest versions, and thus $successors_j$ is always empty set.
The validation of $overwriters_j$ in Silo checks $overwriters_j = \phi$.
If a data item $x_i$ in $rs_j$ is overwritten and there exists a newer version, Silo aborts $t_j$.
To extend Silo to Silo+NWR, we expand the 64-bits per data item transaction-id objects into the 128-bits pivot version objects.
In Silo, each transaction id object stores 1) 32-bits epoch, 2) 31-bits per-epoch version number, and 3) 1-bit exclusive lock flag.
We expand these objects into 128-bits and add mRS and mWS into the rest 64-bits.

\textbf{TicToc.}
This optimistic protocol adopts a more sophisticated validation algorithm for $overwriters_j$ than Silo.
TicToc adds read-timestamps into Silo's transaction id objects.
Even if a version $x_i$ in $rs_j$ is overwritten (it means $overwriters_j \ne \phi$), TicToc validates read-timestamps to commit $t_j$.
TicToc generates the same version order as Silo; TicToc acquires exclusive lockings for all $x_j$ in $ws_j$, and thus each version $x_j$ is the latest version of $x$.
To extend TicToc to TicToc+NWR, we compress the original encoding of the 64-bits transaction-id objects into 32-bits and store the epoch number into the rest 32-bits.

\textbf{MVTO.}
We implemented MVTO based on Cicada \cite{Lim2017Cicada:Transactions};
we apply the per-thread distributed timestamp generation (multi-clock), read sets for optimistic multi-versioning, and rapid garbage collection.
Unfortunately, Cicada's multi-clock does not ensure strict serializability.
It ensures only causal consistency \cite{Lim2017Cicada:Transactions}.
To ensure strict serializability, we apply the epoch-based group commit to both MVTO and MVTO+NWR, by expanding the 64-bits timestamp objects into the 128-bits pivot version objects with epoch numbers.
Note that MVTO aborts a transaction if its write set is invalid, in contrast with the fact that Silo and TicToc validate their read set.
Since the protocol of MVTO ensures that all transactions in $overwriters_j$ cannot be reachable into the running transaction $t_j$, MVTO executes only the validation for $successors_j$.
Hence, we do not have the reusable implementation of validation for $overwriters_j$ within MVTO.
To validate $overwriters_j$, we use the Silo's tiny implementation.
When a version $x_i$ in $rs_j$ is not the latest version of the data item, then MVTO+NWR determines that MVSG might be cyclic and cannot commit $t_j$ with $\ll^{NWR}$.

\subsubsection{Benchmarks}
We select two benchmarks for performance evaluation: YCSB \cite{Cooper2010BenchmarkingYCSB} and TPC-C \cite{10.5555/1946050.1946051}.

\textbf{YCSB.}
This workload generator is representative of current large-scale on-line benchmarks.
YCSB provides various workloads, and we choose YCSB-A (write-intensive) and YCSB-B (read-mostly) for benchmarking.
Since the original YCSB does not support a transaction with multiple operations,
we implemented a YCSB-like workload generator within our prototype.
In our implementation, each transaction accesses four data items chosen randomly based on Zipfian distribution with a parameter $\theta$.
Each data item has a single primary key and an 8-bytes additional column.
We populate our prototype as a single table with 100K data items.
That is, our prototype is always fitted in memory for all single-version protocols.
In this workload, we show how efficient NWR is for contended workloads.

\textbf{TPC-C.}
This workload is the classical industry-standard benchmark for evaluating the performance of transaction processing.
It consists of six tables and five transactions that simulate an information system of a wholesale store.
TPC-C is the write-intensive and low contention workload.
A key point of the TPC-C benchmark is that there exist only blind-inserts.
All other write operations in TPC-C are read-modify-write; they must read the latest versions to write the next versions and always fail to validate SR-Rule.
Therefore, we cannot omit any write operations in TPC-C.
In this workload, we focus on how much performance drops by our extension.

\subsection{YCSB Benchmark Results}
\label{sec:ycsb}

\subsubsection{YCSB-A: read-write mixed workload}
YCSB-A defines the proportion of operations as 50\% read and 50\% blind-write.
A key observation of this workload is that a substantial proportion of write operations are omittable.
In YCSB, all data items must be initialized before experiments \cite{Cooper2010BenchmarkingYCSB}; thus, all write operations are blind-writes.
That is, the validation failure of $\ll^{NWR}$ is rarely happen in NWR-extended protocols.
Once the latest versions for each 40ms epoch are updated, other write operations do not violate ST-Rule.
Besides, for SR-Rule, it is rare to fail the validation of $successors_j$.
This is because blind-writes reduce the number of edges in MVSG with $\ll^{NWR}$; since a transaction $t_j$ trying to omit $w_j(x_j)$ does not have read any $x_k$ in $\{x\}$, there is no cyclic path $t_j \VOWW{} t_{pv} \WR{} t_j$ for data item $x$.  
Thus, NWR-extended protocols are expected to generate a tremendous number of omittable write operations and improve performance.

\textbf{Efficiency of NWR-extended protocols.}
Figure \ref{fig:ycsb-a-high} shows the results on the YCSB-A workload with high contention rate ($\theta = 0.9$).
Prior work shows that recent protocols suffer from this write contended workloads \cite{Wu2017AnControl, Yu2014StaringCores, Kim2016Ermia:Workloads, Fan:2019:OVG:3342263.3360357}.
We later show the performance results in the same workload with variable contention rates.
With up to 21 threads, the throughput of all protocols increases as the number of threads increases.
Beyond 41 threads, coordination such as exclusive locking hurts the performances of baseline protocols.
This is because our environment has only 36 cores per single CPU socket.
Cache-coherence traffic occurred by exclusive locking degrades their throughput significantly.
However, the performances of NWR-extended protocols in YCSB-A are not limited.
When the number of thread increased, NWR-extended protocols outperform originals.
In the best case at 144 threads, the throughput of NWR-extended algorithms is more than 11x better than originals.

\begin{figure}[t]
  \centerline{
    \subfloat{
      \includegraphics[width=0.5\textwidth]{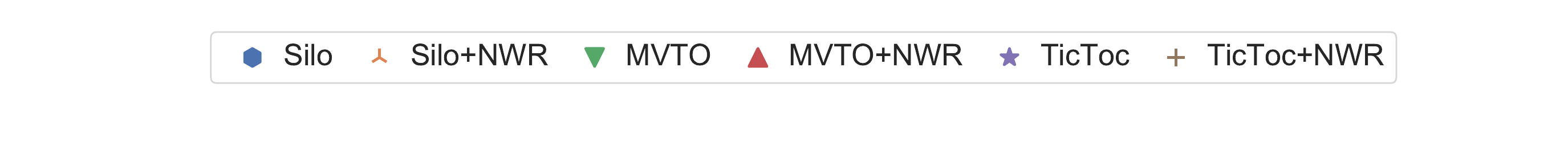}
    }
  }
  \vspace{-10pt}
  \addtocounter{subfigure}{-1}
  \centerline{
    \subfloat{
      \includegraphics[width=0.40\textwidth]{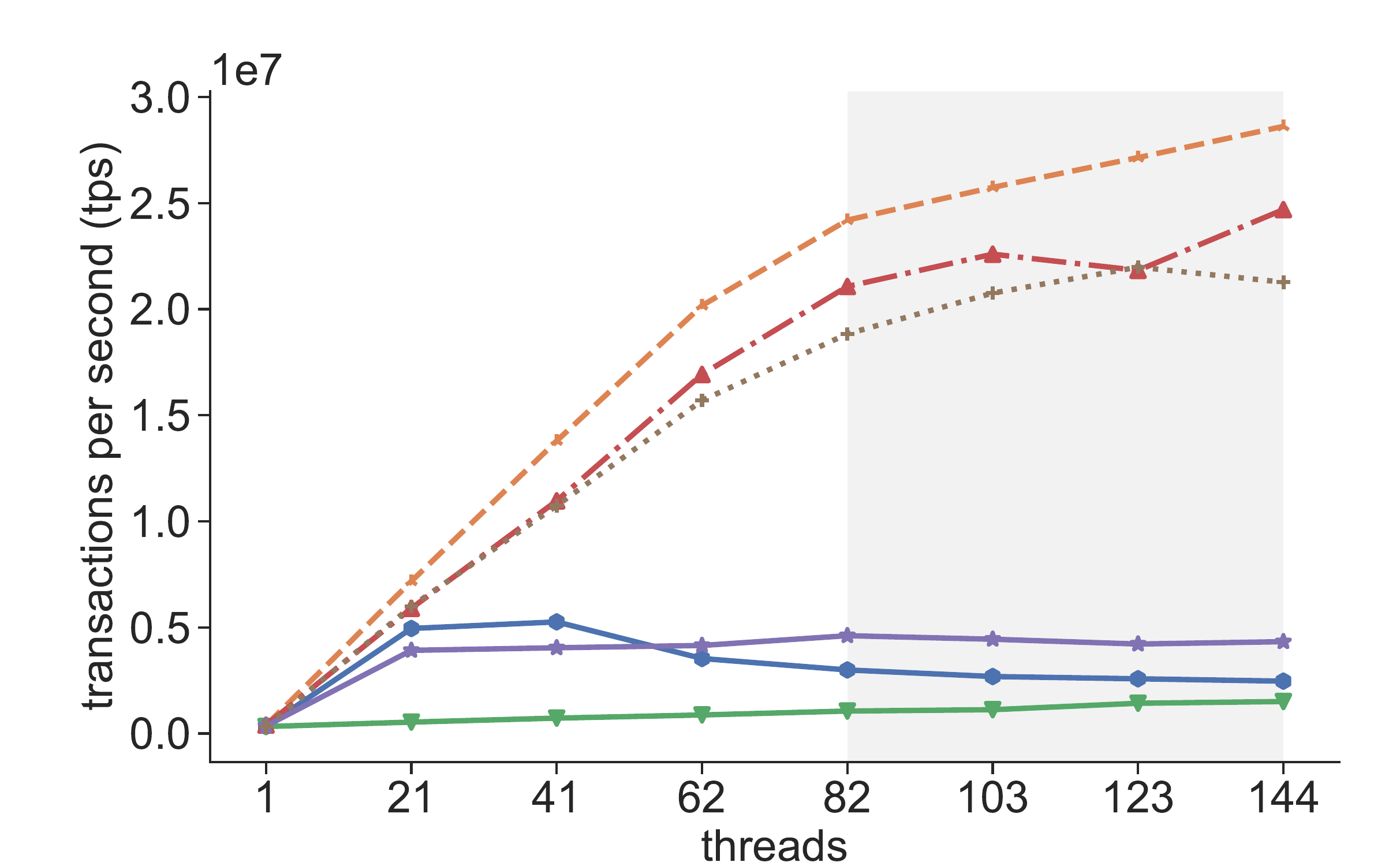}
    }
  }
  \vspace{-6pt}
  \caption{YCSB-A Benchmark results with $\theta=0.9$}
  \label{fig:ycsb-a-high}
\end{figure}

\begin{table}[t]
  \begin{tabular}{|l|c|c|}
    \hline
    Protocols & Commit ratio & Commit ratio with \tiny{$\ll^{NWR}$} \\
    \hline \hline
    Silo+NWR & 94 & 85 \\
    TicToc+NWR & 98 & 65 \\
    MVTO+NWR & 98 & 72  \\
    \hline
  \end{tabular}
  \caption{Commit ratio (\%) of the additional version order $\ll^{NWR}$, generated by NWR-extended protocols, on YCSB-A at 144 threads.}
  \label{tab:commit_ratio}
\end{table}

\begin{figure*}[t]
  \vspace{-10pt}
    \subfloat[
    At 1 thread
    ]{
      \includegraphics[width=0.42\textwidth]{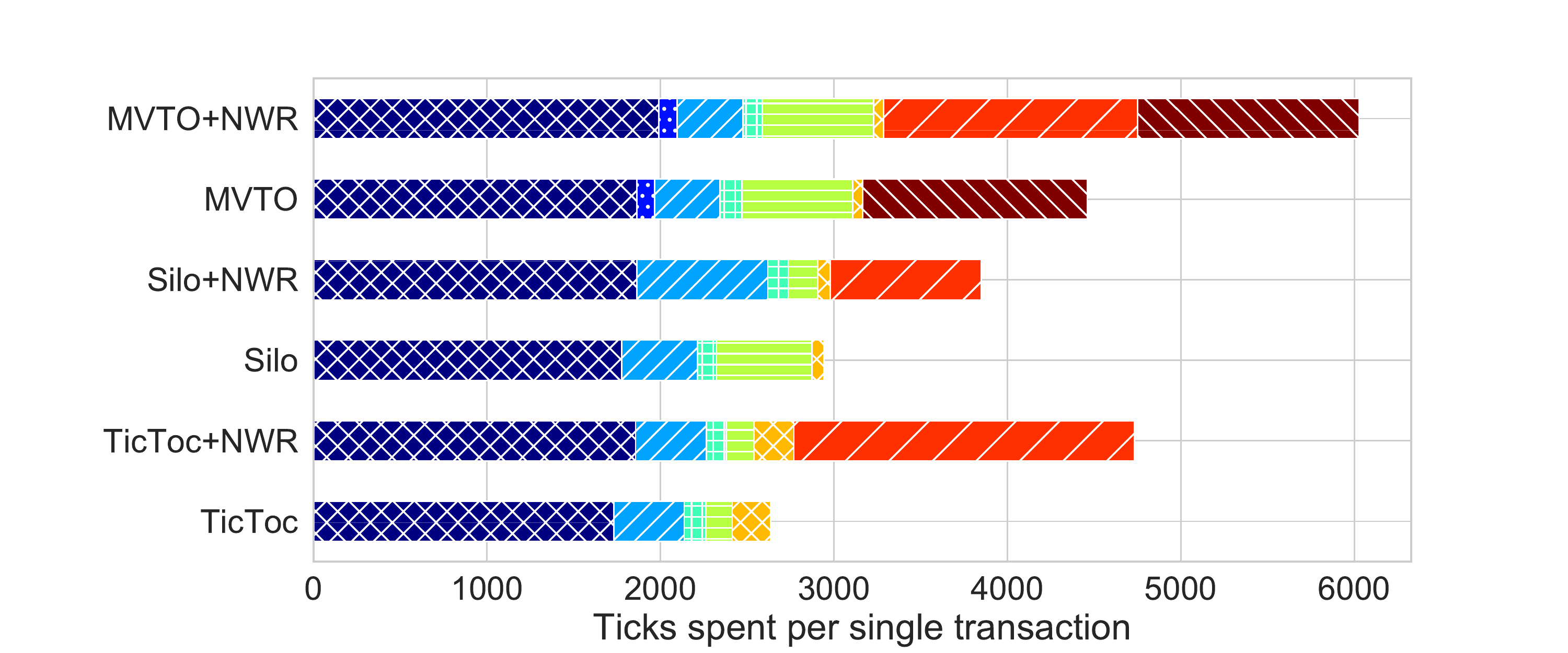}
      \label{fig:ycsb-breakdown-a-1}
    }
    \hspace{-10pt}
    \subfloat[
     At 144 threads
     ]{
      \includegraphics[width=0.42\textwidth]{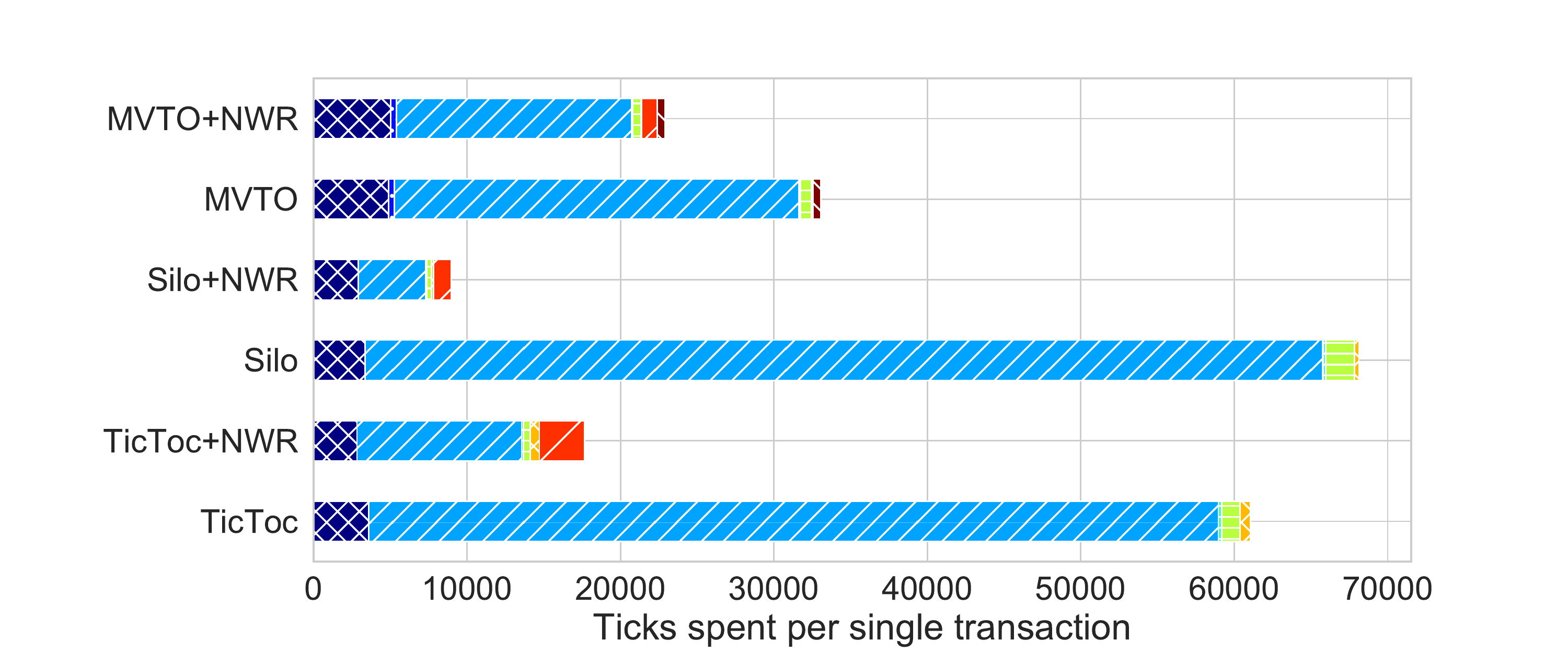}
      \label{fig:ycsb-breakdown-a-144}
    }
    \hspace{-24pt}
    \subfloat{
      \includegraphics[width=0.17\textwidth]{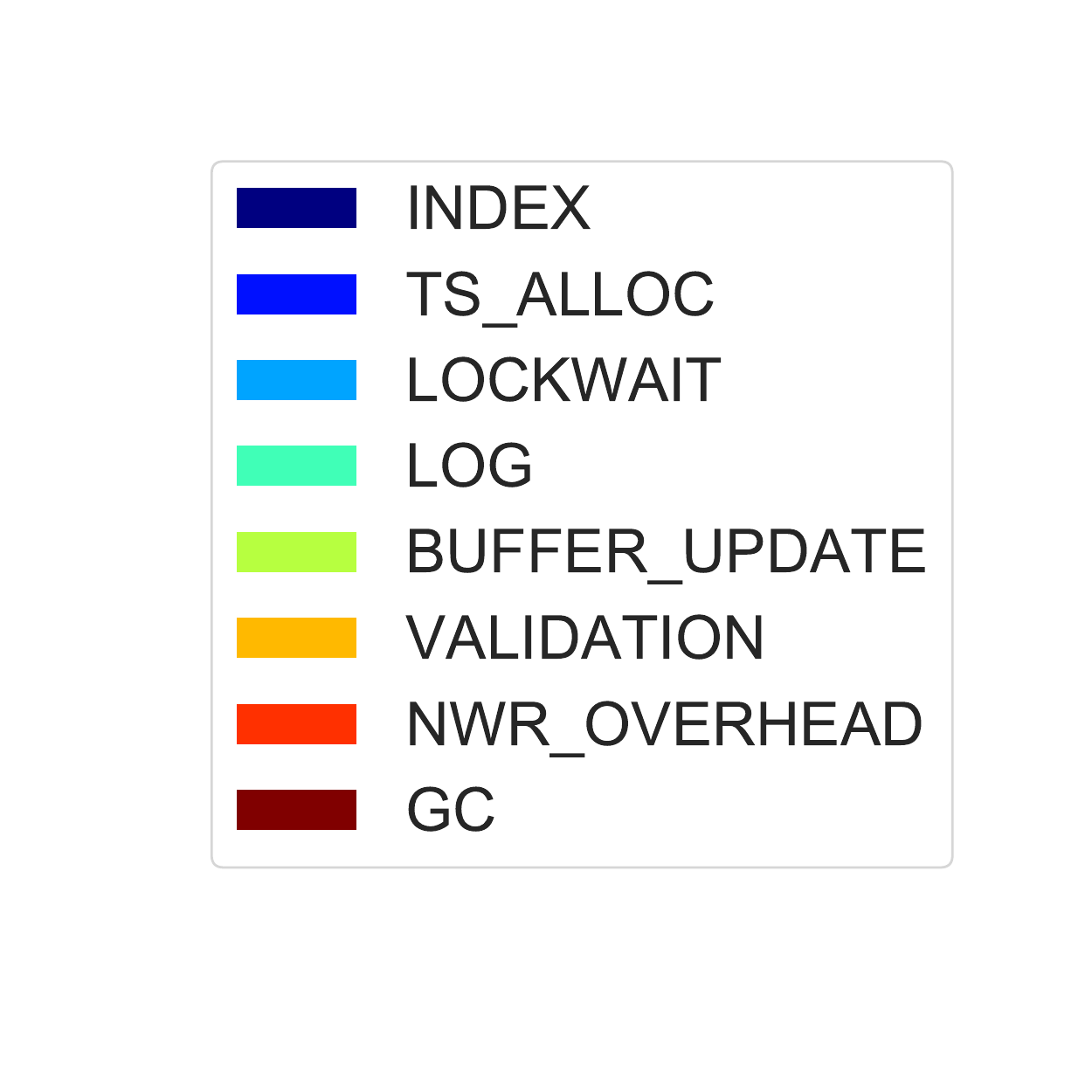}
      \label{fig:ycsb-breakdown-a-144}
    }
  
  \caption{Runtime breakdowns of Figure \ref{fig:ycsb-a-high}}
  \label{fig:ycsb-breakdown}
  \end{figure*}

\begin{figure}[t]
  \centerline{
    \subfloat{
      \includegraphics[width=0.5\textwidth]{benchmark_figures/ICDE/build/legend.pdf}
    }
  }
  \vspace{-10pt}
  \addtocounter{subfigure}{-1}
  \centerline{
    \subfloat{
      \includegraphics[width=0.40\textwidth]{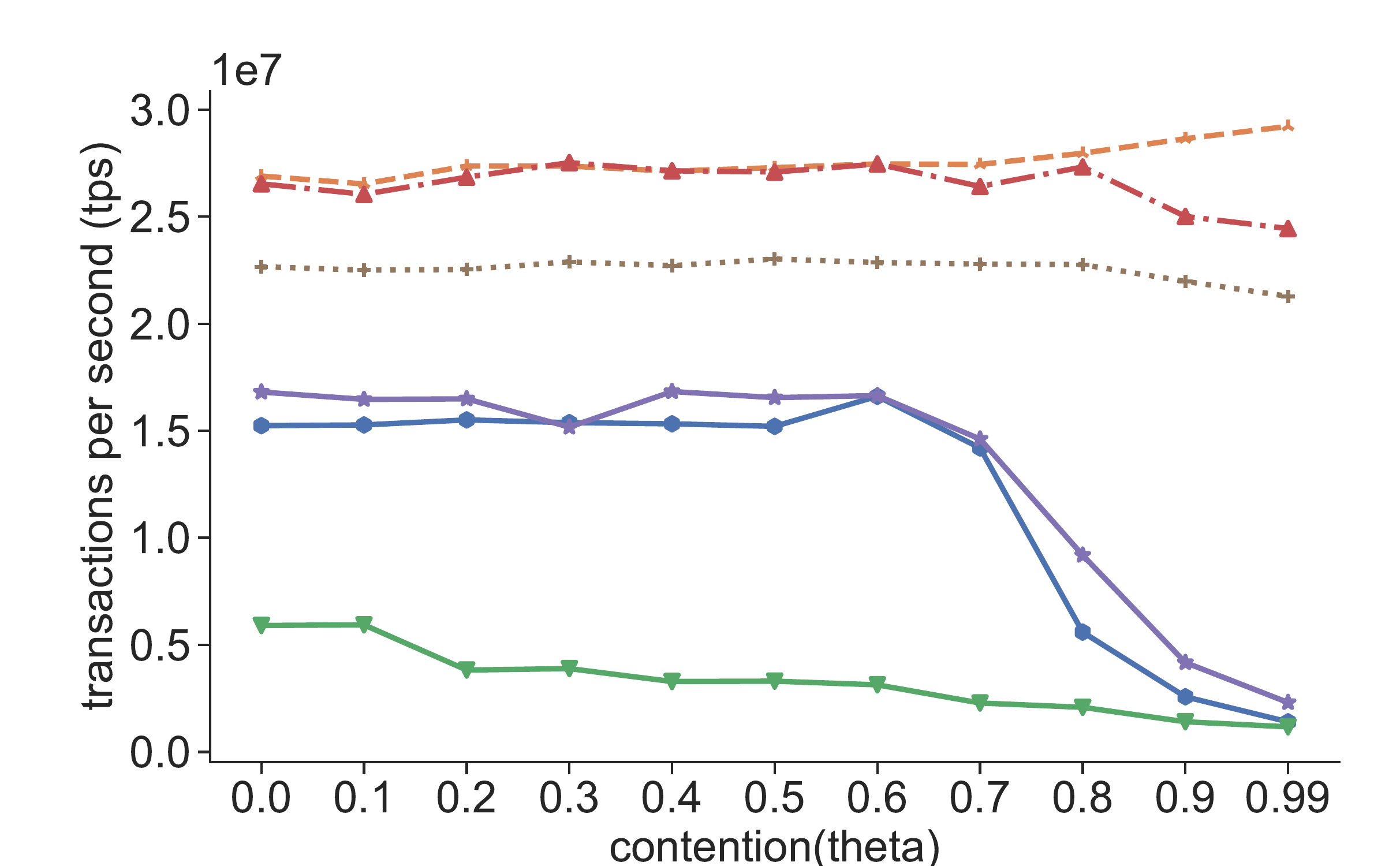}
    }
  }
  \vspace{-6pt}
  \caption{YCSB-A Benchmark results with variable $\theta$}
  \label{fig:ycsb-variable-contention}
\end{figure}

\textbf{Differences between baseline protocols.}
The results for YCSB-A also shows that TicToc+NWR and MVTO+NWR do not scale, unlike Silo+NWR.
The main reason for the performance gap is that TicToc and MVTO allow \textit{reading old versions}.
When a transaction reads a non-latest version that has already overwritten, Silo must abort the transaction, but TicToc and MVTO do not.
They may commit the transaction if their validation pass, and thus mRS in the pivot version storage may be updated with lower version numbers at any time.
Recall that the validation (4) in Algorithm \ref{alg:successors_validation} fails when a version number $y_g$ in mRS is lower.
Therefore, reading old versions decreases the commit ratio of NWR-extended protocols with $\ll^{NWR}$.
Table \ref{tab:commit_ratio} lists the commit ratio of the additional version order $\ll^{NWR}$ for each NWR-extended protocols.
Although the commit ratio of TicToc+NWR and MVTO+NWR are better than Silo+NWR, they sometimes give up to omit write operations in the validation of $\ll^{NWR}$.
In contrast, Silo+NWR efficiently commits transactions with $\ll^{NWR}$ and omits write operations.
It is known that the validation algorithm of Silo has higher false abort rates \cite{Yu2016Tictoc:Control, Yuan2016BCC:Databases, Ding2018ImprovingOptimistic}; however, it is favorable for our approach with $\ll^{NWR}$.

\textbf{Runtime breakdown.}
Figure \ref{fig:ycsb-breakdown-a-1} and \ref{fig:ycsb-breakdown-a-144} show the runtime breakdown of committed transactions at a single and 144 threads on the YCSB-A with $\theta = 0.9$, respectively.
``NWR\_OVERHEAD'' indicates the additional components of our approach described in Section \ref{sec:control_flow}.
At a single thread, the top consumer of CPU ticks is INDEX for all protocols, but the overhead of NWR is not negligible for all NWR-extended protocols.
Compared to Silo+NWR, TicToc+NWR and MVTO+NWR spend more CPU ticks for the additional components since they adopt the different approaches for validation of $overwriters_j$.
Recall that TicToc and MVTO allow reading old versions, but Silo does not; thus, they spend more ticks on updating mRS for each pivot version.
At 144 threads, the primary consumers of CPU-ticks shift to LOCKWAIT, and the overhead of NWR becomes negligibly small.
NWR-extended protocols dramatically reduce the CPU ticks wasting for locking by generating omittable write operations.
Note that MVTO, which shows the lowest performance in Figure \ref{fig:ycsb-a-high}, spent fewer CPU ticks per single transaction than Silo and TicToc.
This is because our runtime breakdown measures only the CPU ticks on concurrency control protocols for committed transactions.
Other overheads such as malloc for creating new versions or lock-wait for aborted transactions are not shown in this figure, and they are current bottlenecks.
We remain cache-efficient optimization for MVTO and MVTO+NWR as future work.

\begin{figure}[t]
  \centerline{
    \subfloat{
      \includegraphics[width=0.5\textwidth]{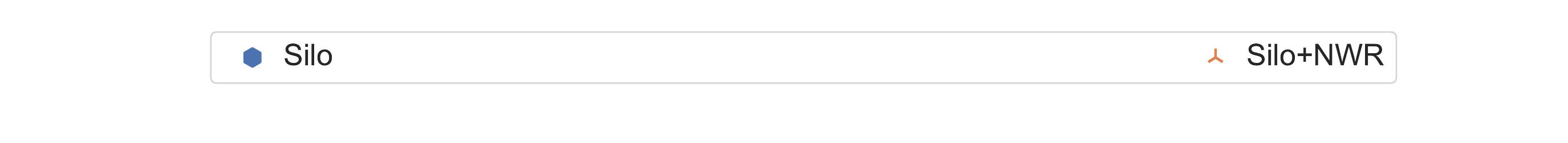}
    }
  }
  \vspace{-10pt}
  \addtocounter{subfigure}{-1}
  \centerline{
    \subfloat{
      \includegraphics[width=0.40\textwidth]{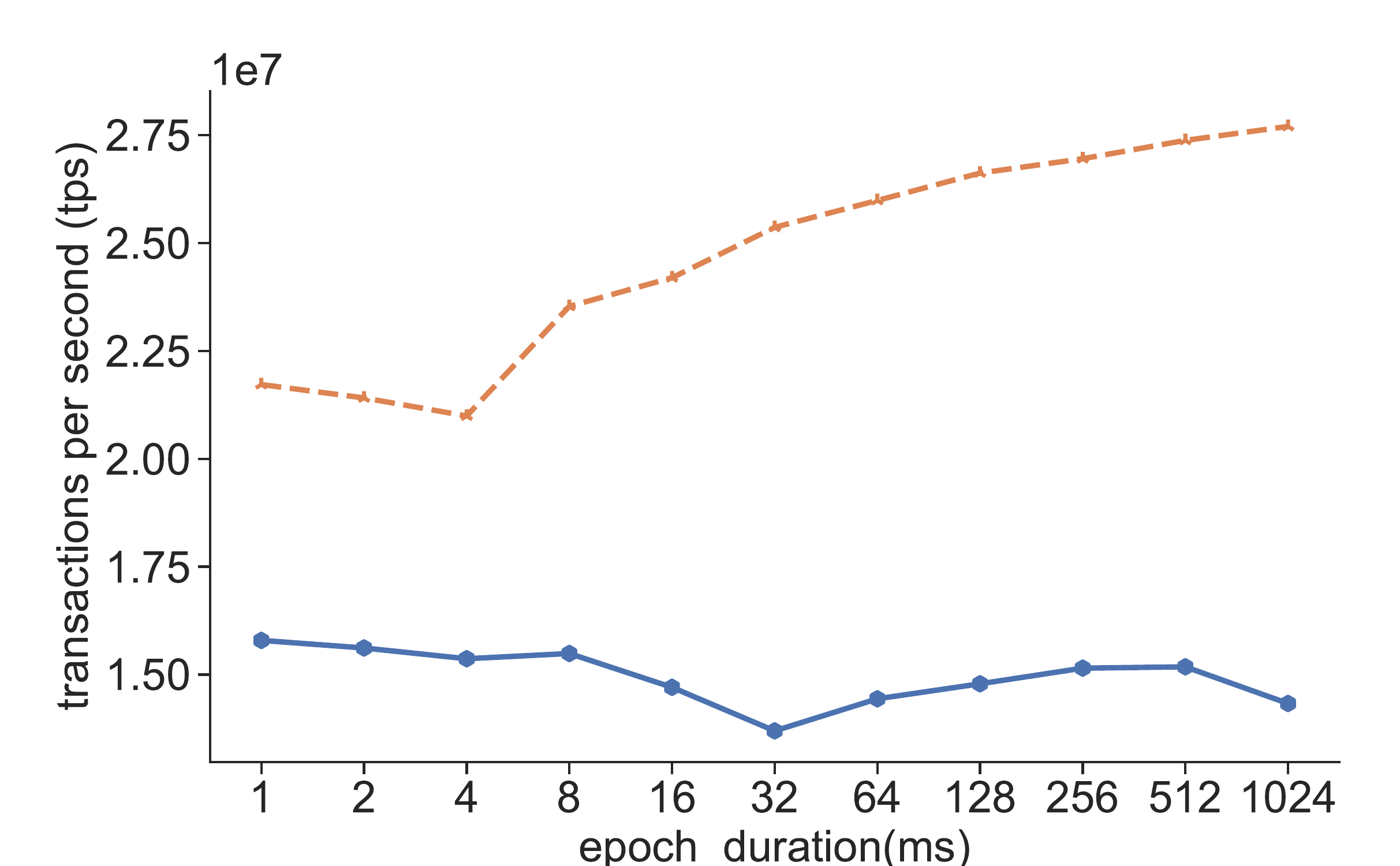}
    }
  }
  \vspace{-6pt}
  \caption{YCSB-A Benchmark results with variable epoch durations}
  \label{fig:ycsb-variable-epoch}
\end{figure}

\subsubsection{YCSB-A with variable contention rates.}
Next, we experiment YCSB-A with variable contention rates.
Figure \ref{fig:ycsb-variable-contention} shows the results at 144 threads.
Once the contention rate exceeds 0.7, the throughputs of Silo+NWR and TicToc+NWR drops due to exclusive locking for write operations.
MVTO suffers from cache-coherence traffic caused by multiversion storage.
In MVTO, it is unnecessary to acquire exclusive locking for write operations, but instead, memory allocation for creating new versions is needed, and it limits the throughput.
At the contention rate of 0.99, every baseline protocol presents almost the same throughput.
Previous studies show the same trend with this result \cite{Fan:2019:OVG:3342263.3360357, Yu2014StaringCores, Wu2017AnControl}.
However, NWR-extended protocols outperform the originals and do not degrade the throughput when the contention rate increased.
Furthermore, Silo+NWR improves performance at the contention rate of 0.99.
In this setting, transactions request to read or write to the same data items, and Silo+NWR successes to omit almost all write operations.
Therefore, almost all read operations receive the same versions, which are already stored in the CPU caches.
Compared to Silo+NWR, MVTO+NWR and TicToc+NWR do not degrade but cannot improve performance in this setting, since they allow \textit{reading old versions} as has been noted above.

\subsubsection{YCSB-A with variable epoch durations.}
As has been shown in Section \ref{sec:implementation}, NWR-extended protocols ensure ST-Rule by using epoch-based group commits.
Transactions in the same epoch are concurrent and satisfy ST-Rule for any version order.
Hence, the longer epoch duration will make more transactions concurrent and decrease the validation failure of ST-Rule.
To investigate the effect of epoch duration, we experiment YCSB-A with variable epoch duration.
Figure \ref{fig:ycsb-variable-epoch} depicts the throughputs of Silo and Silo+NWR.
The original Silo shows unstable performance as the epoch duration increases.
This is because the longer epoch duration causes larger commit buffers for each epoch.
In contrast, the throughput of Silo+NWR grows as the epoch duration increases since the ST-Rule validation failures decrease.
Thus, we can improve performance by increasing the epoch duration as long as applications allowed.

\subsubsection{YCSB-B: read-mostly workload}
Figure \ref{fig:ycsb-b-high} shows the performance results on the read-mostly workload.
YCSB-B is expected as not suitable for our approach since it defines the proportion of read operation as 95\%; that is, concurrent write operations into the same data item are rarely observed.
However, NWR-extended protocols show comparable performances with the originals, and surprisingly, Silo+NWR outperforms Silo.
This gain comes from the observation that omitting write operations have a good effect on the other transactions.
In this case, Silo+NWR replaces overwriting operations to omittable write operations, and it prevents the validation failure of other transactions.
Figure \ref{fig:ycsb-b-high-abort} shows the number of aborts in the same experiment.
Here we see Silo+NWR reduces the abort rate of the original Silo, and it provides the performance improvement in this workload.


\begin{figure}[t]
  \centerline{
    \subfloat{
      \includegraphics[width=0.5\textwidth]{benchmark_figures/ICDE/build/legend.pdf}
    }
  }
  \vspace{-10pt}
  \addtocounter{subfigure}{-1}
  \centerline{
    \subfloat{
      \includegraphics[width=0.40\textwidth]{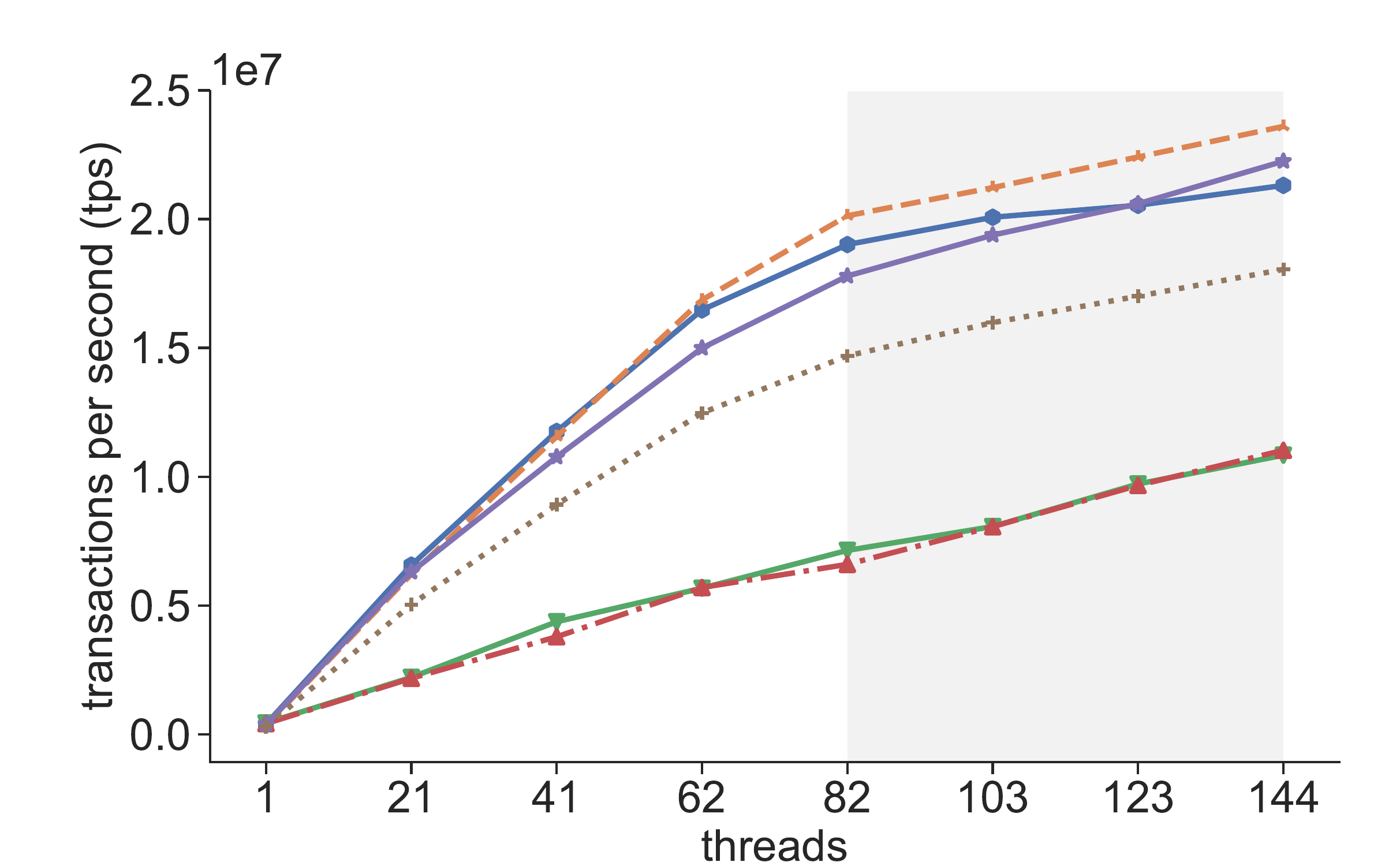}
    }
  }
  \vspace{-6pt}
  \caption{Throughput in YCSB-B with $\theta=0.9$}
  \label{fig:ycsb-b-high}
\end{figure}

\begin{figure}[t]
  \centerline{
    \subfloat{
      \includegraphics[width=0.5\textwidth]{benchmark_figures/ICDE/build/legend_only_silo.pdf}
    }
  }
  \vspace{-10pt}
  \addtocounter{subfigure}{-1}
  \centerline{
    \subfloat{
      \includegraphics[width=0.40\textwidth]{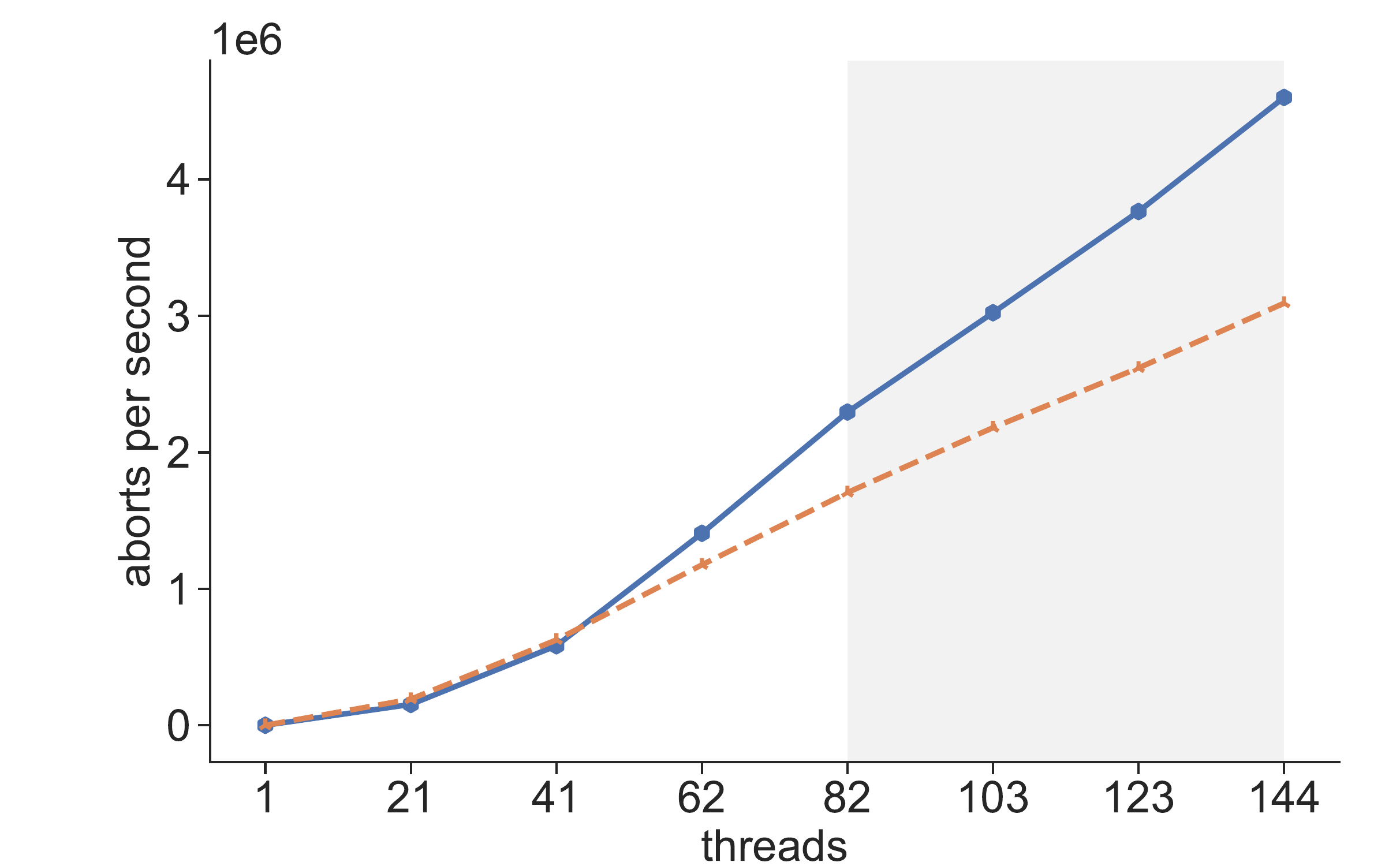}
      }
  }
  \vspace{-6pt}
  \caption{Aborts of Silo and Silo+NWR in YCSB-B}
  \label{fig:ycsb-b-high-abort}
\end{figure}


\subsection{TPC-C Benchmark Results}
\label{sec:tpc-c}

We implement TPC-C in our prototype, the same as the original specifications.
Recall that all five queries in TPC-C do not contain blind-writes without inserting queries.
It means that there are no omittable write operations in this workload.
We consider emerging workloads contains tremendous blind-writes; however, we experiment on this benchmark to illustrate the low-overhead property of our NWR-extended protocols.

\textbf{Overheads of the pivot version storage.}
To better understand the low-overhead property of NWR-extended protocols, we run the TPC-C benchmark with a single warehouse.
This high-contended scenario presents the worst case for NWR since the pivot version storage for each data item will be updated frequently but never be utilized.
Figure \ref{fig:tpcc-tps} shows the achieved throughputs as the number of threads increased.
It can be seen that Silo+NWR and TicToc+NWR's overheads are negligibly small.
The performance degradations in these two NWR-extended protocols are less than 10\%, and they achieved comparable performances with originals.
However, MVTO+NWR presents about 0.84x lower throughput than the original MVTO.
This is because the additional components in MVTO+NWR require traversing version lists twice, and it incurs the performance penalty from cache-coherence traffic. 
As we have been noted in Section \ref{sec:implementation}, MVTO+NWR uses Silo's anti-dependency validation; MVTO+NWR validates that there exists some newer version for each $t_k$ in $rs_j$.
This validation requires to traverse the version list again for each $t_k$.
Hence, MVTO+NWR needs to traverse a version list twice per single read operation, once for a read operation, and once for validating $overwriters_j$.
This additional traversing is harmful to the CPU cache, as the number of elements in a version list increases.
In TPC-C, the queries new-order and payment frequently perform read-modify-write into the warehouse and district tables, and the longer version-lists degrades the performance of MVTO+NWR.


\begin{figure}[t]
  \centerline{
    \subfloat{
      \includegraphics[width=0.5\textwidth]{benchmark_figures/ICDE/build/legend.pdf}
    }
  }
  \vspace{-10pt}
  \addtocounter{subfigure}{-1}
  \centerline{
    \subfloat{
      \includegraphics[width=0.40\textwidth]{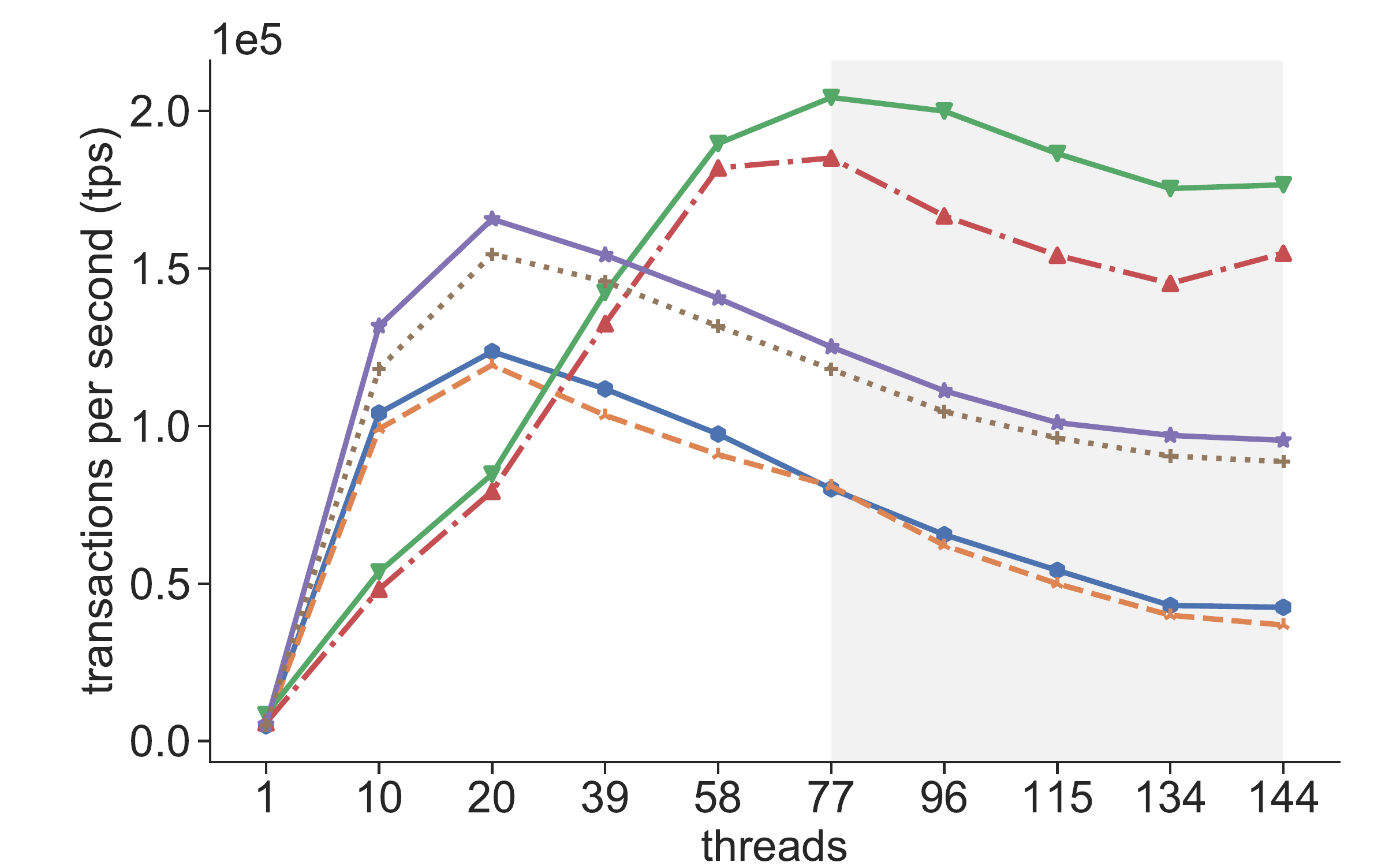}
    }
  }
  \vspace{-6pt}
  \caption{TPC-C benchmark results with a single warehouse}
  \label{fig:tpcc-tps}
\end{figure}



\section{Related Work}

\textbf{Thomas Write Rule.}
As we have been discussed in Section \ref{sec:formal_aspects_of_iw}, TWR allows T/O to omit some write operations.
Compared to NWR, the rule of TWR is more restrictive since TWR requires T/O's timestamps.
TWR validates only a single version order generated by the timestamps and allows omitting only when the timestamp order is satisfied with the rule of TWR.
NWR is the more elastic rule.
It allows any protocol to validate any number of version orders, and it does not depend on any baseline implementation.

\textbf{Lazy transaction execution.} 
Another example of the protocol which allows generating omittable write operations is in deterministic databases \cite{Ren2014AnSystems,Thomson2012Calvin:Systems}.
Deterministic databases adopt the centralized approach of concurrency control; centralized transaction managers collect transactions and separate them into some batches.
Faleiro et al. \cite{Faleiro2014LazySystems} proposed the approach of \textit{lazy transaction execution}.
It enables deterministic databases to omit some write operations.
By delaying the execution of write operations, deterministic databases execute only blind updates, which writes the latest version in each batch.
Lazy transaction execution makes other write operations omittable.
This approach is similar to our approach of NWR-extended protocols, in terms of grouping transactions; both approaches separate transactions and execute only the latest versions. 
Compared to lazy evaluation, the rule of NWR and the approach of validating the additional version order is more applicable for any protocol, since the rule of NWR does not require the centralized transaction manager.
NWR is useful for both deterministic and non-deterministic databases.

\textbf{Multiversion concurrency control protocols.}
Multiversion concurrency control (MVCC) protocols \cite{Lim2017Cicada:Transactions, Fekete2005MakingSerializable, Larson2011High-performanceDatabases} can hold multiple versions for each data item.
Multiversion read avoids the highly abort rates of single-version protocols, especially in workloads that include long-transactions \cite{Yu2014StaringCores,Wu2017AnControl,Kim2016Ermia:Workloads}.
MVSR is widely used to guarantee the correctness of MVCC, and thus all MVCC protocols can in theory generate multiple version orders.
However, existing multiversion protocols \cite{Larson2011High-performanceDatabases,Lim2017Cicada:Transactions,Kim2016Ermia:Workloads,Ports2012SerializablePostgreSQL,Kemper2011HyPer:Snapshots} generate only a single version order.
This is because the computational cost of the decision problem to find a suitable version order from all possible orders is NP-Complete \cite{Bernstein1983MultiversionAlgorithms,Papadimitriou1982OnVersions}.
Furthermore, the approach of validating multiple version orders requests protocols to verify multiple MVSGes, and it causes enormous overheads. 
NWR and NWR-extended protocols efficiently reduce the computational cost of validating multiple version orders.
NWR-extended protocols add only a single version order, which is ideal for omitting; thus, they validate only two MVSGes.
Besides, PV-Rule of NWR defines that the two MVSGes are isomorphic if a running transaction $t_j$ aborts.
Therefore, additional validation is required only for a connected subgraph which includes $t_j$.
NWR-extended protocols achieve to validate multiple version orders efficiently; they maintain only a single version order and validate two almost the same MVSGes. 

\section{Conclusion}
In this paper we have formalized the notion of omitting allowed by T/O with TWR, and we have presented the Non-visible write rule (NWR) for any concurrency control protocols.
With NWR, any protocol can in theory be extensible to NWR-extended protocols, which validates the additional version order to omit write operations with preserving both strict serializability and recoverability.
Especially in the workloads that contain blind-updates, our approach of validating the additional version order improves the performance efficiently.
We implemented three NWR-extended protocols, Silo+NWR, TicToc+NWR, and MVTO+NWR, to evaluate the performance gain of our approach.
Experimental results demonstrate that Silo+NWR outperforms the original Silo by more than 11x in the highly contended YCSB-A; meanwhile, every NWR-extended protocol has comparable performance with the original, even in the worst-case on TPC-C.

\bibliographystyle{plain}

\begin{thebibliography}{10}

\bibitem{10.5555/1946050.1946051}
Transaction processing performance council (tpc): Standard specification. tpc
  benchmark c version 5.11.0 (2010).

\bibitem{2000:GIL:846219.847380}
Generalized isolation level definitions.
\newblock In {\em Proceedings of the 16th International Conference on Data
  Engineering}, ICDE '00, pages 67--, Washington, DC, USA, 2000. IEEE Computer
  Society.

\bibitem{10.1145/223784.223785}
Hal Berenson, Phil Bernstein, Jim Gray, Jim Melton, Elizabeth O’Neil, and
  Patrick O’Neil.
\newblock A critique of ansi sql isolation levels.
\newblock In {\em Proceedings of the 1995 ACM SIGMOD International Conference
  on Management of Data}, SIGMOD ’95, page 1–10, New York, NY, USA, 1995.
  Association for Computing Machinery.

\bibitem{Bernstein1982ConcurrencySystems}
Philip~A Bernstein and Nathan Goodman.
\newblock {Concurrency control algorithms for multiversion database systems}.
\newblock In {\em Proceedings of the first ACM SIGACT-SIGOPS symposium on
  Principles of distributed computing}, pages 209--215, 1982.

\bibitem{Bernstein1983MultiversionAlgorithms}
Philip~A Bernstein and Nathan Goodman.
\newblock {Multiversion concurrency control—theory and algorithms}.
\newblock {\em ACM Transactions on Database Systems (TODS)}, 8(4):465--483,
  1983.

\bibitem{Bernstein1987ConcurrencySystems}
Philip~A Bernstein, Vassos Hadzilacos, and Nathan Goodman.
\newblock {\em {Concurrency control and recovery in database systems}}.
\newblock Addison-Wesley Pub. Co. Inc., Reading, MA, 1987.

\bibitem{Bernstein:2009:PTP:1208930}
Philip~A. Bernstein and Eric Newcomer.
\newblock {\em Principles of Transaction Processing}.
\newblock Morgan Kaufmann Publishers Inc., San Francisco, CA, USA, 2nd edition,
  2009.

\bibitem{Chandramouli2018FASTER:Updates}
Badrish Chandramouli, Guna Prasaad, Donald Kossmann, Justin Levandoski, James
  Hunter, and Mike Barnett.
\newblock {FASTER: A concurrent key-value store with in-place updates}.
\newblock In {\em Proceedings of the 2018 International Conference on
  Management of Data}, pages 275--290, 2018.

\bibitem{Cooper2010BenchmarkingYCSB}
Brian~F Cooper, Adam Silberstein, Erwin Tam, Raghu Ramakrishnan, and Russell
  Sears.
\newblock {Benchmarking Cloud Serving Systems with YCSB}.
\newblock In {\em Proceedings of the 1st ACM Symposium on Cloud Computing},
  SoCC '10, pages 143--154, New York, NY, USA, 2010. ACM.

\bibitem{Ding2018ImprovingOptimistic}
Bailu Ding, Lucja Kot, and Johannes Gehrke.
\newblock Improving optimistic concurrency control through transaction batching
  and operation reordering.
\newblock {\em Proc. VLDB Endow.}, 12(2):169–182, October 2018.

\bibitem{Durner2019NoFalseNegatives}
Dominik Durner and Thomas Neumann.
\newblock No false negatives: Accepting all useful schedules in a fast
  serializable many-core system.
\newblock In {\em 35th {IEEE} International Conference on Data Engineering,
  {ICDE} 2019, Macao, China, April 8-11, 2019}, pages 734--745, 2019.

\bibitem{Faleiro2014LazySystems}
Jose~M Faleiro, Alexander Thomson, and Daniel~J Abadi.
\newblock {Lazy evaluation of transactions in database systems}.
\newblock In {\em Proceedings of the 2014 ACM SIGMOD international conference
  on Management of data}, pages 15--26, 2014.

\bibitem{Fan:2019:OVG:3342263.3360357}
Hua Fan and Wojciech Golab.
\newblock Ocean vista: Gossip-based visibility control for speedy
  geo-distributed transactions.
\newblock {\em Proc. VLDB Endow.}, 12(11):1471--1484, July 2019.

\bibitem{Fekete2005MakingSerializable}
Alan Fekete, Dimitrios Liarokapis, Elizabeth O'Neil, Patrick O'Neil, and Dennis
  Shasha.
\newblock {Making snapshot isolation serializable}.
\newblock {\em ACM Transactions on Database Systems (TODS)}, 30(2):492--528,
  2005.

\bibitem{Gray1992TransactionTechniques}
Jim Gray and Andreas Reuter.
\newblock {\em {Transaction processing: concepts and techniques}}.
\newblock Elsevier, 1992.

\bibitem{Guide2011IntelManual}
Part Guide.
\newblock {Intel{\textregistered} 64 and ia-32 architectures software
  developer’s manual}.
\newblock {\em Volume 3B: System programming Guide, Part}, 2, 2011.

\bibitem{Hadzilacos1988ASystems}
Vassos Hadzilacos.
\newblock {A theory of reliability in database systems}.
\newblock {\em Journal of the ACM (JACM)}, 35(1):121--145, 1988.

\bibitem{Harizopoulos2008OLTPThere}
Stavros Harizopoulos, Daniel~J. Abadi, Samuel Madden, and Michael Stonebraker.
\newblock {OLTP through the looking glass, and what we found there}.
\newblock In {\em Proceedings of the 2008 ACM SIGMOD international conference
  on Management of data - SIGMOD '08}, page 981, New York, New York, USA, 2008.
  ACM Press.

\bibitem{Herlihy1990Linearizability:Objects}
Maurice~P Herlihy and Jeannette~M Wing.
\newblock {Linearizability: A correctness condition for concurrent objects}.
\newblock {\em ACM Transactions on Programming Languages and Systems (TOPLAS)},
  12(3):463--492, 1990.

\bibitem{Johnson2010Aether:Logging}
Ryan Johnson, Ippokratis Pandis, Radu Stoica, Manos Athanassoulis, and
  Anastasia Ailamaki.
\newblock {Aether: a scalable approach to logging}.
\newblock {\em Proceedings of the VLDB Endowment}, 3(1-2):681--692, 2010.

\bibitem{Kemper2011HyPer:Snapshots}
Alfons Kemper and Thomas Neumann.
\newblock {HyPer: A hybrid OLTP{\&}OLAP main memory database system based on
  virtual memory snapshots}.
\newblock In {\em 2011 IEEE 27th International Conference on Data Engineering},
  pages 195--206, 2011.

\bibitem{Kim2016Ermia:Workloads}
Kangnyeon Kim, Tianzheng Wang, Ryan Johnson, and Ippokratis Pandis.
\newblock {Ermia: Fast memory-optimized database system for heterogeneous
  workloads}.
\newblock In {\em Proceedings of the 2016 International Conference on
  Management of Data}, pages 1675--1687, 2016.

\bibitem{Kung1981OnControl}
Hsiang-Tsung Kung and John~T Robinson.
\newblock {On optimistic methods for concurrency control}.
\newblock {\em ACM Transactions on Database Systems (TODS)}, 6(2):213--226,
  1981.

\bibitem{Larson2011High-performanceDatabases}
{Larson, Per-Åke and Blanas, Spyros and Diaconu, Cristian and Freedman, Craig
  and Patel, Jignesh M and Zwilling, Mike}.
\newblock {High-performance concurrency control mechanisms for main-memory
  databases}.
\newblock {\em Proceedings of the VLDB Endowment}, 5(4):298--309, 2011.

\bibitem{Lim2017Cicada:Transactions}
Hyeontaek Lim, Michael Kaminsky, and David~G Andersen.
\newblock {Cicada: Dependably fast multi-core in-memory transactions}.
\newblock In {\em Proceedings of the 2017 ACM International Conference on
  Management of Data}, pages 21--35, 2017.

\bibitem{Mao2012CacheStorage}
Yandong Mao, Eddie Kohler, and Robert~Tappan Morris.
\newblock {Cache craftiness for fast multicore key-value storage}.
\newblock In {\em Proceedings of the 7th ACM european conference on Computer
  Systems}, pages 183--196, 2012.

\bibitem{Papadimitriou1986TheControl}
Christos Papadimitriou.
\newblock {\em {The theory of database concurrency control}}.
\newblock Computer Science Press Inc., Rockville, MD, 1986.

\bibitem{Papadimitriou1982OnVersions}
Christos~H Papadimitriou and Paris~C Kanellakis.
\newblock {On Concurrency Control by Multiple Versions}.
\newblock In {\em Proceedings of the 1st ACM SIGACT-SIGMOD Symposium on
  Principles of Database Systems}, PODS '82, pages 76--82, New York, NY, USA,
  1982. ACM.

\bibitem{Ports2012SerializablePostgreSQL}
Dan R~K Ports and Kevin Grittner.
\newblock {Serializable snapshot isolation in PostgreSQL}.
\newblock {\em Proceedings of the VLDB Endowment}, 5(12):1850--1861, 2012.

\bibitem{Reed1978NamingSystem.}
David~Patrick Reed.
\newblock {\em {Naming and synchronization in a decentralized computer
  system.}}
\newblock PhD thesis, Massachusetts Institute of Technology, 1978.

\bibitem{10.14778/3342263.3342647}
Kun Ren, Dennis Li, and Daniel~J. Abadi.
\newblock Slog: Serializable, low-latency, geo-replicated transactions.
\newblock {\em Proc. VLDB Endow.}, 12(11):1747–1761, July 2019.

\bibitem{Ren2014AnSystems}
Kun Ren, Alexander Thomson, and Daniel~J Abadi.
\newblock {An evaluation of the advantages and disadvantages of deterministic
  database systems}.
\newblock {\em Proceedings of the VLDB Endowment}, 7(10):821--832, 2014.

\bibitem{Technology2012AMD64Date}
Amd Technology.
\newblock {AMD64 Technology AMD64 Architecture Programmer’s Manual Volume 3:
  General-Purpose and System Instructions Publication No. Revision Date}, 2012.

\bibitem{Thomas1977ABases}
Robert~H Thomas.
\newblock {A majority consensus approach to concurrency control for multiple
  copy data bases}.
\newblock Technical report, BOLT BERANEK AND NEWMAN INC CAMBRIDGE MA, 1977.

\bibitem{Thomson2012Calvin:Systems}
Alexander Thomson, Thaddeus Diamond, Shu-Chun Weng, Kun Ren, Philip Shao, and
  Daniel~J Abadi.
\newblock {Calvin: fast distributed transactions for partitioned database
  systems}.
\newblock In {\em Proceedings of the 2012 ACM SIGMOD International Conference
  on Management of Data}, pages 1--12, 2012.

\bibitem{Tu2013SpeedyDatabases}
Stephen Tu, Wenting Zheng, Eddie Kohler, Barbara Liskov, and Samuel Madden.
\newblock {Speedy transactions in multicore in-memory databases}.
\newblock In {\em Proceedings of the Twenty-Fourth ACM Symposium on Operating
  Systems Principles}, pages 18--32, 2013.

\bibitem{Wang2017EfficientlySerializable}
Tianzheng Wang, Ryan Johnson, Alan Fekete, and Ippokratis Pandis.
\newblock {Efficiently making (almost) any concurrency control mechanism
  serializable}.
\newblock {\em The VLDB Journal}, 26(4):537--562, 2017.

\bibitem{Weikum2001TransactionalRecovery}
Gerhard Weikum and Gottfried Vossen.
\newblock {\em {Transactional information systems: theory, algorithms, and the
  practice of concurrency control and recovery}}.
\newblock Elsevier, 2001.

\bibitem{Wu2017AnControl}
Yingjun Wu, Joy Arulraj, Jiexi Lin, Ran Xian, and Andrew Pavlo.
\newblock {An Empirical Evaluation of In-memory Multi-version Concurrency
  Control}.
\newblock {\em Proc. VLDB Endow.}, 10(7):781--792, 3 2017.

\bibitem{Yu2014StaringCores}
Xiangyao Yu, George Bezerra, Andrew Pavlo, Srinivas Devadas, and Michael
  Stonebraker.
\newblock {Staring into the abyss: An evaluation of concurrency control with
  one thousand cores}.
\newblock {\em Proceedings of the VLDB Endowment}, 8(3):209--220, 2014.

\bibitem{Yu2016Tictoc:Control}
Xiangyao Yu, Andrew Pavlo, Daniel Sanchez, and Srinivas Devadas.
\newblock {Tictoc: Time traveling optimistic concurrency control}.
\newblock In {\em Proceedings of the 2016 International Conference on
  Management of Data}, pages 1629--1642, 2016.

\bibitem{Yuan2016BCC:Databases}
Yuan Yuan, Kaibo Wang, Rubao Lee, Xiaoning Ding, Jing Xing, Spyros Blanas, and
  Xiaodong Zhang.
\newblock {BCC: reducing false aborts in optimistic concurrency control with
  low cost for in-memory databases}.
\newblock {\em Proceedings of the VLDB Endowment}, 9(6):504--515, 2016.

\bibitem{Zheng2014FastParallelism}
Wenting Zheng, Stephen Tu, Eddie Kohler, and Barbara Liskov.
\newblock {Fast Databases with Fast Durability and Recovery Through Multicore
  Parallelism}.
\newblock In {\em 11th USENIX Symposium on Operating Systems Design and
  Implementation (OSDI 14)}, pages 465--477, Broomfield, CO, 2014. USENIX
  Association.

\end{thebibliography}

\appendix
\section{Formal Proof}
\label{sec:proof}

We need the following lemmas to prove Theorem \ref{theo:NWR_is_correct}.
Each lemma corresponds to serializability, strict serializability, and recoverability, respectively.

\begin{lemma}[Serializable]
\label{theo:NWR_is_sr}
For a triple $(S, \ll, t_j)$ which satisfies NWR,
$S \cup \{c_j\}$ is serializable.
\end{lemma}

\begin{proof}
To prove by contradiction,
we assume that $MVSG(S \cup \{c_j\}, \ll^{NWR})$ for any $\ll^{NWR}$ is not acyclic.
By PV-Rule, $MVSG(S, \ll)$ and $MVSG(S, \ll^{NWR})$ are isomorphic.
To form the cycle, there exists some paths $t_j \to t_k \to ..., \to t_j$ for some $t_k$.
By SR-Rule, there exists no edge $t_k \to t_j$ for all $t_k \in RN(t_j)$.
It implies the contradiction.
\end{proof}




\textbf{Strict serializability.}
The definition of strict serializability is following \cite{Herlihy1990Linearizability:Objects}:
\begin{quote}
A history is serializable if it is equivalent to one in which transactions appear to execute sequentially, i.e., without interleaving. A (partial) precedence order can be defined on non-overlapping pairs of transactions in the obvious way. A history is strictly serializable if the transactions’ order in the sequential history is compatible with their precedence order. 
\end{quote}
As has been noted, a schedule $S$ is serializable iff there exists a version order $\ll$ such that $MVSG(S, \ll)$ is acyclic.
Strict serializability is the condition of the order $M$, a total order of transactions generated by topologically-sorting of acyclic MVSG.
We define the formal definition\footnotemark[3] of strict serializability for a schedule $S$ and a total order of transactions $M$ as following:
$\forall t_i \forall t_k \forall p \forall q(((p \in t_i \land q \in t_k) \land p <_S q) \Rightarrow t_i <_M t_k)$ where $<_M$ in $M$ and $<_S$ in $S$.

\begin{lemma}[Strictly serializable]
\label{theo:linearizable}
For a triple $(S, \ll, t_j)$ which satisfies NWR,
$S \cup \{c_j\}$ is strictly serializable.
\end{lemma}

\begin{proof}
By Lemma \ref{theo:NWR_is_sr}, $MVSG(S \cup \{c_j\}, \ll^{NWR})$ is acyclic for some $\ll^{NWR}$.
Let $M$ be a total order of $trans(CP(S \cup \{c_j\}))$ generated by commit-ordered first topological-sorting for MVSG.
To prove by contradiction, we assume that $S\cup \{c_j\}$ is not strictly serializable.
By assumption, there exists some $t_i$ in $trans(CP(S))$ such that (a) $\forall p \forall q(((p \in t_i \land q \in t_j) \land p <_S q)$ but (b) $t_j <_M t_i$.
(b) implies that MVSG includes the path $\{t_j \to ..., \to t_i\}$.
However, by ST-Rule, there exists some $q$ in $t_j$ such that $q <_S c_k$ for all reachable transactions $t_k$ in $RN(t_j)$.
That is, if (a) holds for some transaction $t_i$, (b) does not hold; there always exists $q$ in $t_j$ such that $q <_S c_i$.
It implies the contradiction.
\end{proof}

\footnotetext[3]{We define that the operation order in a given schedule reflects the wall-clock order.}

\textbf{Recoverability.}
The formal definition of recoverability is the following \cite{Weikum2001TransactionalRecovery}: 
$$\forall t_i \forall t_j ((t_i, t_j \in trans(S) \land t_i \ne t_j) \Rightarrow (r_j(x_i) \in t_j \Rightarrow c_i <_S c_j)$$
\begin{lemma}[Recoverable]
    \label{theo:recoverable}
For a triple $(S, \ll, t_j)$ which satisfies NWR,
$S \cup \{c_j\}$ is recoverable.
\end{lemma}
\begin{proof}
To prove by contradiction, we assume that $S\cup \{c_j\}$ is not recoverable.
Since $S$ is recoverable, there exists $x_i$ such that $c_i,a_i \notin S$ and $x_i$ in $rs_j$.
By RC-Rule, it does not hold.
\end{proof}

Finally, we now prove Theorem \ref{theo:NWR_is_correct}.
\begin{proof}
    \label{proof:NWR_is_correct}
    By Lemma \ref{theo:linearizable}, $S \cup \{c_j\}$ is strictly serializable.
    By Lemma \ref{theo:recoverable}, $S \cup \{c_j\}$ is recoverable.
\end{proof}

\balance
\end{document}